\documentclass[a4paper,11pt]{article}

\usepackage{jheppub,amsmath,dsfont,mathtools,braket,bm} 

\usepackage[T1]{fontenc} 

\usepackage{amsthm}

\newcommand{\be}{\begin{equation}}
\newcommand{\ee}{\end{equation}}
\newcommand{\N}{\mathcal{N}}

\makeatletter
\newcommand{\superimpose}[2]{%
  {\ooalign{$#1\@firstoftwo#2$\cr\hfil$#1\@secondoftwo#2$\hfil\cr}}}
\makeatother
             
\newtheorem{prop}{Proposition}[section]

\newtheorem{definition}{Definition}[section]

\title{\boldmath Gauge $\times$ Gauge $=$ Gravity on Homogeneous Spaces using Tensor Convolutions}

\author[a]{L. Borsten,}
\author[b,1]{I. Jubb\note{Corresponding author.},}
\author[]{V. Makwana,}
\author[c]{S. Nagy}

\affiliation[a]{Maxwell Institute for Mathematical Sciences,\\
Department of Mathematics, Heriot-Watt University,\\
 Colin Maclaurin Building, Riccarton, Edinburgh EH14 4AS, United Kingdom}
\affiliation[b]{School of Theoretical Physics, Dublin Institute for Advanced Studies,\\
10 Burlington Road, Dublin 4, Ireland}
\affiliation[c]{Queen Mary University of London, 327 Mile End Road, London E1 4NS, United Kingdom}

\emailAdd{l.borsten@hw.ac.uk}
\emailAdd{ijubb@stp.dias.ie}
\emailAdd{visheshmakwana@gmail.com}
\emailAdd{s.nagy@qmul.ac.uk}

\abstract{A definition  of a convolution of tensor fields on group manifolds is given, which  is then generalised to generic  homogeneous spaces. This is applied to the product of gauge fields in the context of `gravity $=$ gauge $\times$ gauge'. In particular, it is shown that the linear  Becchi--Rouet--Stora--Tyutin (BRST) gauge transformations of two Yang-Mills gauge fields generate the linear BRST diffeomorphism transformations of the graviton. This facilitates the definition of the `gauge $\times$ gauge' convolution product on, for example, the static Einstein universe, and more generally for ultrastatic spacetimes with compact spatial slices.}

\begin{document} 
\maketitle

\section{Introduction}
We introduce a convolution of tensor fields on group  manifolds, which to the best of our knowledge has not been treated in the literature to date.  This is generalised to homogeneous spaces, extending  the special case of $S^2$ given in \cite{Borsten:2019prq}. These convolutions are then  applied to the notion of `gravity $=$ gauge $\times$ gauge'. In particular, it is shown that the symmetric convolution of two gauge potentials yields a graviton. The gauge symmetry BRST transformations    generate via the convolution the diffeomorphism BRST transformations of the graviton, to linear order. This allows us to apply the field theoretic  `gravity $=$ gauge $\times$ gauge' construction of \cite{Anastasiou:2014qba, Borsten:2017jpt, Anastasiou:2018rdx, Zoccali:2018pty, Borsten:2019prq, Borsten:2020xbt,Luna:2020adi,Cardoso:2016amd,Cardoso:2016ngt} on curved spacetime backgrounds, such as the $D=4$ spacetime dimensional Einstein universe. We also note that the convolution of tensor fields on homogeneous spaces is of intrisic interest and may have broader applications, cf. for example \cite{Dokmanic:2010, Chakraborty:2018h, Cohen:2018general} and the references therein. 

\paragraph{Gravity $=$ gauge $\times$ gauge} Let us briefly review the theme of `gravity $=$ gauge $\times$ gauge'. For  details and more complete references see the reviews \cite{Bern:2019prr, Borsten:2020bgv}. From the Kawai-Lewellen-Tye (KLT) relations of string theory \cite{Kawai:1985xq}, which relate closed string tree amplitudes  to the sums of products of open string tree amplitudes, we learn that the tree amplitudes of perturbatively quantised $\N=0$ supergravity (Einstein-Hilbert gravity, coupled to a dilaton $\varphi$ and a Kalb-Ramond 2-form $B$)  are the `square' of Yang-Mills amplitudes.  While the KLT relations are intrinsically tree level, it was shown in certain examples that  this relation could be extended to low loop orders \cite{Bern:1993wt,Bern:1998ug}. This programme was dramatically advanced with the Bern-Carrasco-Johansson colour/kinematic (CK) duality conjecture \cite{Bern:2008qj}: the gluon amplitudes can be cast in a form such that their `kinematic numerators' (Lorentz invariant polynomials of momenta and polarisation tensors) obey the same relations as their  `colour numerators' (polynomials of the gauge group structure constants).  CK duality has been shown to hold at tree level from a number of points of view \cite{Stieberger:2009hq,BjerrumBohr:2009rd,BjerrumBohr:2010hn, Feng:2010my,  Chen:2011jxa, Mafra:2011kj, Du:2016tbc, Mizera:2019blq, Reiterer:2019dys}. It remains conjectural at loop level where it quickly becomes difficult to test \cite{Bern:2017yxu, Bern:2017ucb}, although there are numerous highly non-trivial examples \cite{Bern:2010ue,Bern:2010tq,Carrasco:2011mn,Bern:2011rj,BoucherVeronneau:2011qv,Bern:2012cd,Bern:2012gh, Bern:2012uf,Du:2012mt,Yuan:2012rg,Bern:2013uka, Boels:2013bi,Bern:2013yya,Bern:2013qca, Bern:2014sna, Mafra:2015mja, Johansson:2017bfl,  Bern:2017ucb}. 

If  CK duality is satisfied by a Yang-Mills amplitude integrand, then its \emph{double-copy} is an amplitude integrand of $\N=0$ supergravity \cite{Bern:2010ue, Bern:2010yg}. This suggests a fundamental relationship between gauge theory and gravity, at least perturbatively, and reveals new features and puzzles regarding the properties of perturbative quantum gravity \cite{Bern:2012cd,Bern:2012uf,Bern:2014sna}. This motivates  some natural questions. Does CK duality and/or the double-copy hold to all orders in perturbation theory; is the double-copy special to amplitudes; can it be taken beyond perturbation theory; are there applications beyond the computation of gravity amplitudes; what are implications for quantum gravity? There are many approaches that one might take in addressing these challenges. For instance, there is an ambitwistor string approach to CK duality and the double-copy, related to the scattering equation formalism for the double-copy \cite{Cachazo:2013iea, Cachazo:2014xea, Mason:2013sva, Adamo:2013tsa}.   Beyond perturbation theory, for particular (e.g.~Kerr-Schild) spacetimes there is a non-perturbative  classical double-copy of Yang-Mills solutions \cite{Monteiro:2014cda,Luna:2015paa, Luna:2016due,Alawadhi:2019urr,Banerjee:2019saj}, with an elegant variation relating the square of Yang-Mills field strengths to the Weyl tensor \cite{Luna:2018dpt,Alawadhi:2019urr, White:2020sfn,  Monteiro:2020plf, Chacon:2021wbr}. As to applications, there is a vigorous and promising programme to bend  amplitudes and the double-copy to the problem of classical black hole  scattering in the context of gravity wave astronomy \cite{Bern:2019crd, Bern:2019nnu,Bern:2020buy,Bern:2020gjj,Bern:2020uwk,Bern:2021dqo}. CK duality and the double-copy even find applications in fluid dynamics \cite{Cheung:2020djz}. In \cite{Campiglia:2021srh}, a double copy for asymptotic symmetries led to the discovery of a new symmetry for self-dual YM at null infinity, identified as the single copy of gravitational superrotations.

Two ideas that are central to our present discussion are (i) that the double-copy can be applied off-mass-shell to the fields of two gauge theories \cite{Bern:2010yg, Anastasiou:2014qba} and (ii) that this should be extended to all the fields of the BRST complex, including the longitudinal and ghost modes \cite{Anastasiou:2018rdx}. Regarding  (i), it was shown that  CK duality for physical gluon tree-level amplitudes can  be made manifest order-by-order at the Lagrangian level \cite{Bern:2010yg, Tolotti:2013caa}. One can then double-copy the CK dual Lagrangian itself, yielding a theory that produces the correct tree level scattering amplitudes for $\N=0$ supergravity, as was shown  to six points in \cite{Bern:2010yg}. In \cite{Anastasiou:2014qba, Anastasiou:2018rdx} an \emph{a priori} off-shell convolution product of the fields, including those of the BRST complex, of two independent gauge theories was introduced. For two pure BRST Yang-Mills theories in flat space this yields the BRST complex and linear BRST transformations of perturbative $\N=0$ supergravity \cite{Anastasiou:2018rdx}. Combined with the Lagrangian double-copy \cite{Bern:2010yg} the pure Einstein-Hilbert action to cubic order was derived from that of Yang-Mills theory using the convolution product \cite{Borsten:2020xbt}\footnote{See also \cite{Ferrero:2020vww} for the use of the convolution product to construct the $\mathcal{N}=0$ supergravity.}. The chief advantages of the convolution, including the extra BRST fields, are that (i) the ghost sector  allows the dilaton to be  truncated without imposing further constraints on the graviton \cite{LopesCardoso:2018xes, Borsten:2020xbt} and (ii) the Yang-Mills gauge choice  determines the $\N=0$ supergravity gauge choice to linear order, removing the ambiguities inherent in, for example, the double-copy of gauge dependent solutions. In \cite{Borsten:2020zgj,Borsten:2021hua} it was shown that tree-level CK duality holds for amplitudes involving all states of the BRST Fock space, including the unphysical longitudinal gluon and ghost modes, and that this can be made manifest in a purely cubic Yang-Mills BRST-Lagrangian\footnote{Although tree-level CK duality holds for all states of the extended BRST Fock space, this does not necessarily imply loop-level CK duality. However,  the gluon loop amplitude integrands computed  with the Feynman diagrams of the manifest BRST-CK dual Yang-Mills action of \cite{Borsten:2020zgj,Borsten:2021hua} provide `almost BCJ numerators' that double-copy correctly into the loop amplitude integrands of $\N=0$ supergravity \cite{Borsten:2020zgj,Borsten:2021hua}.}. It was then shown that its Lagrangian double-copy yields a BRST-Lagrangian that is pertubatively quantum equivalent to $\N=0$ supergravity to all orders, tree and loop \cite{Borsten:2020zgj, Borsten:2021hua}. A direct corollary is that all tree and loop Yang-Mills amplitude integrands can be written in a form that double-copies correctly, i.e.~yields a bona fide $\N=0$ supergravity amplitude.  Let us emphasise that key to all of the preceding discussion was the derivation of the linear BRST operator  of the double copy theory from  the BRST operators of the  gauge theory factors \cite{Anastasiou:2018rdx}. 

\paragraph{Curved backgrounds} Everything  till  now has assumed perturbation theory around a flat spacetime (and gluon) background. It is natural to ask what of the double-copy survives in curved spaces, or at least in some suitable class of curved spaces. There are various possibilities.   One could consider a flat gluon background on a curved spacetime, a curved gluon background on flat spacetime or both non-trivial gluon and spacetime backgrounds. For example, CK duality and the double-copy  in a curved `sandwich' plane-wave gluon background was considered in \cite{Adamo:2017nia,Adamo:2020qru}. The Kerr-Schild version of the double copy also admits a formulation on curved backgrounds \cite{Bahjat-Abbas:2017htu,Alkac:2021bav}. There have also been generalisations to AdS and conformal correlators \cite{Farrow:2018yni,Lipstein:2019mpu,Armstrong:2020woi,Albayrak:2020fyp,alday2021gluonp}. 

On the other hand, in \cite{Borsten:2019prq} a convolution product for the BRST complexes of two gauge theories on a spatial sphere (with  trivial gluon background) was introduced. Trivially including a time dimension, this provided a convolution product in the $D=3$ Einstein universe. It was shown that the linear BRST transformations correctly double-copy to those of perturbative $\N=0$ supergravity  on an $D=3$ Einstein universe background. 

It was noted in \cite{Borsten:2019prq} that the convolution of tensor fields on $S^2$ relied on rather generic  properties of homogeneous spaces. Encouraged  by this observation, in the present contribution we consider the convolution product for tensor fields on  general  Riemannian  homogeneous spaces. Trivially including a time dimension facilitates the double-copy of the BRST transformations on a broad class of spacetimes such as the $D=4$ Einstein universe.

\paragraph{Structure} We proceed as follows. In \autoref{brst} we recall the essentials of the BRST  formalism in the context of Yang-Mills theory and $\N=0$ supergravity. For complete details, see \cite{Borsten:2021hua}. In \autoref{brstprod} we summarise what is required of the convolution product with respect to the goal of generating the linear diffeomorphism BRST transformations from the product of the Yang-Mills BRST transformations. Specialising to ultra-static spacetimes, reduces the problem to defining a convolution on Riemannian  homogeneous manifolds, as discussed in \autoref{spatial}.

Our construction of a tensor field convolution on homogeneous spaces  relies upon a simpler convolution on group manifolds $G$, which may be regarded as special class of homogeneous manifolds $G\cong (G\times G)/G$. Accordingly, we first formulate  in \autoref{group} the tensor convolution on compact Lie groups. The features of group manifolds essential to the formulation of our tensor field convolution are reviewed in \autoref{groupprelim}. The convolution is introduced in \autoref{convgroup} and  it properties under differentiation are determined in \autoref{diffgroup}. The convolution of functions on group manifolds is straightforward and well-known. It essentially relies on the existence of a $G$-bi-invariant measure (the Haar measure). The obvious obstruction to extending this construction to tensor fields is the need to compare tensors at different points on $G$. However,  a basic property of Lie groups is that the left and right multiplication diffeomorphisms generate a set of left and right invariant basis vectors on $T_gG$ fo all $g\in G$,  starting from some choice of basis of vectors at the identity $T_eG$.   Using this observation we introduce maps that generate left- and right-invariant vector fields from any vector at any point. This allows us to translate any vector field in a unique manner, facilitating a convolution of tensor fields as defined in \autoref{convgroup}. Having introduced the convolution, its properties  are developed. In particular, it is shown in \autoref{diffgroup} that the symmetrised covariant derivative has the same properties as the flat space derivative acting on the flat space convolution when act on the convolution of functions with 1-forms. This implies that the BRST transformations of the graviton  follow from those of the Yang-Mills gauge potentials to linear order.   Finally, the generalisation of the familiar Convolution Theorem is given in \autoref{convthrm} 

Having treated the special case of Lie groups, we turn our attention to Riemannian  homogeneous manifolds $M\cong G/H$. In \autoref{homrev} we review the basics of homogeneous spaces. The group manifold convolution is then lifted to one defined on any Riemannian  homogeneous space in \autoref{convhom}. The basic idea is to regard $G$ as the fibre bundle $H\rightarrow G\rightarrow M$, using the projection $\pi: G\rightarrow M$ to  `pullback' the convolution of $p$-form fields on $M$ to the convolution defined on $G$ and then projecting back down to $M$.   The action of the symmetrised covariant derivative on the symmetrised convolution is considered in \autoref{diffhom}  and  shown to have the same properties as the group manifold case, implying that the graviton BRST transformations are generated correctly.

\section{Squaring  BRST}
\subsection{BRST Review}\label{brst}
 The `square' of pure Yang-Mills theory
 \be
S_{\text{YM}}=\frac{1}{2 g_{\rm YM}^2}  \int   {\rm tr} F\wedge  \star F, \qquad F = DA:=dA+A\wedge A 
\ee
ought to correspond to  the    universal Neveu-Schwarz  sector of the $\alpha'\rightarrow 0$ limit of closed string theories, 
\be
S_{\N=0} =\frac{1}{2\kappa^2} \int \star  \left(R -(D-2)e^{\frac{4}{D-2}\varphi}\Lambda \right)-\frac{1}{(D-2)}  d\varphi \wedge \star  d\varphi - \frac{1}{2}e^{-\frac{4}{D-2}\varphi}  H \wedge  \star  H,
\ee
where $2\kappa^2=16\pi G_\text{N}^{(D)}$. Aside from the metric $g$ and cosmological constant $\Lambda$, we have the dilaton $\varphi$ and the Kalb-Ramond  (KR) 2-form $B$ with field strength $H=dB$. This is sometimes referred to as  $\N=0$ supergravity, for short. For a Minkowski background this follows from the relationship between the tree-level BCJ double-copy and the KLT relations of string theory \cite{Kawai:1985xq}. 

We shall be concerned with relating the linearised BRST transformations of Yang-Mills theory to those $\N=0$ supergravity. So, let us briefly review the linearised  BRST  formalism here. Of course, for free gauge and gravity theories, the ghosts decouple and there is no need to pass through BRST. However, it \emph{is} nonetheless important in the context of squaring Yang-Mills theory \cite{Anastasiou:2014qba, Anastasiou:2018rdx, Borsten:2020xbt, Borsten:2020zgj, Borsten:2021hua}.

\paragraph{Yang-Mills theory}
The BRST complex consists of the ghost number ${\sf gh}=0$ gauge potential $A$, its ${\sf gh}=1$ ghost $c$, and the trivial pair of the  ${\sf gh}=0$ Nakanishi-Lautrup auxiliary field $b$ and ${\sf gh}=-1$ antighost $\bar c$. The linearised ${\sf gh}=1$ off-shell nilquadratic $Q_{\rm YM}^{2}=0$ BRST transformations are given by
\begin{subequations}\label{BRST_YM}
\begin{align} \label{BRST_YM}
Q_{\rm YM} A &= dc,\\
Q_{\rm YM} c &= 0,\\
Q_{\rm YM} \bar{c}&= b,\\
Q_{\rm YM} b &=0.
\end{align}
\end{subequations}
The physical states are contained in the cohomology of $Q_{\rm YM}$.

For a given gauge-fixing condition, the linearised $Q_{\rm YM}$-invariant BRST action can be written
\be
S^{\rm lin}_{\text{YM, BRST}}={\rm tr} \int     \left(\frac12 dA \wedge  \star dA  +Q_{\rm YM} \Psi_A\right)
\ee
where $\Psi_A$ is the ghost number $-1$ gauge-fixing fermion
\be
\Psi_A =    {\rm tr}  \bar c (G[A] -  \frac\alpha2 b)
\ee
for gauge-fixing function $G[A]$ with Gaussian width $\alpha\in\mathds{R}$. A typical choice of gauge-fixing are the $R_\alpha$-linear gauges $G[A]=-d^\dagger A = {\rm div} A$. This yields 
\be
S^{\rm lin}_{\text{YM, BRST}}=    \int  {\rm tr} \left( \frac12 dA \wedge  \star dA + \star b G[A] - \star \frac{\alpha}{2} b^2   -   \star  \bar{c} G[dc] \right),
\ee
which upon eliminating $b$ gives 
\be
S^{\rm lin}_{\text{YM, BRST}}=    \int  {\rm tr} \left( \frac12 dA \wedge  \star dA + \star \frac{1}{2\alpha} G[A]^2 -  \star \bar{c}  G[dc]   \right),
\ee
with 
\be
Q_{\rm YM} \bar c = \frac1\alpha G[A]. 
\ee

\paragraph{$\N=0$ supergravity} We consider the linearisation around some arbitrary background metric on a $D=(d+1)$-dimensional Lorentzian manifold $M$. More explicitly, we consider a one-parameter family of metric and dilaton fluctuations, 
\be
g(\kappa) = g + \kappa h +\mathcal{O}(\kappa^2), \qquad \phi(\kappa) = \phi_0 + \kappa \varphi +\mathcal{O}(\kappa^2)
\ee
where $g$ and $\phi_0$ are a background metric and dilaton  solving the  Einstein and scalar equations of motion and for notational convenience we consider $\kappa$ as our parameter. In fact, we shall consider arbitrary background metrics through the inclusion of arbitrary sources, but they will be treated only implicitly.  Then 
\be\label{limN=0}
\begin{split}
S^{\rm lin}_{\N=0} &:=\lim_{\kappa\rightarrow 0} S_{\N=0}(\kappa)\\
& = \int dx^D\sqrt{-g}\left( \mathcal{L}_{\rm FP} -\tfrac{1}{(D-2)}  d\varphi \wedge \star  d\varphi - \tfrac{1}{2} H \wedge  \star  H+\mathcal{L}[g, \phi_0, \Lambda, \varphi, h]\right),
\end{split}
\ee
where $\mathcal{L}_{\rm FP}$ is the Fierz-Pauli action quadratic in $h_{\mu\nu}$ and $\mathcal{L}[g, \phi_0, \Lambda, \varphi, h]$ is linear in the fluctuations $\varphi$ and $h$.

The Fierz-Pauli action has a gauge symmetry (the residue of diffeomorphism invariance upon taking the limit \eqref{limN=0}),
\be
\delta h = \nabla \xi,  
\ee
or in components 
\be
\delta h_{\mu\nu} = 2\nabla_{(\mu} \xi_{\nu)},
\ee
where $\nabla_\mu$ is the covariant derivative with respect to the Levi-Civita connection of the background metric $g$. 

Here we have introduced the coordinate independent symmetrised derivative $\nabla$. For any function $f$ it is defined by  $\nabla f := df \in T^* M$, where $d$ is the exterior derivative, and for any $\omega\in T^* M$ we define it as
\begin{equation}\label{eq:gen_def_symm_cov_deriv}
(\nabla \omega )(X,Y) := (\nabla_{X}\omega )(Y) + (\nabla_{Y}\omega )(X) \;\;\; ,
\end{equation}
for any $X, Y \in TM$. Note that the symmetrised covariant derivative $\nabla$ has the same symbol as the usual covariant derivative $\nabla_X$ with respect to $X \in TM$ ($X^\mu \nabla_\mu$ in components), but any ambiguity between the two can be resolved by the fact that the symmetrised covariant derivative does not take a subscript argument.

We work in Einstein frame so the dilaton is a scalar. The Fierz-Pauli
 BRST complex consists of the ghost number ${\sf gh}=0$ gauge potential $h$, the 1-form diffeomorphism  ${\sf gh}=1$ ghost $\xi$, and its accompanying  1-form trivial pair of the  ${\sf gh}=0$ Nakanishi-Lautrup auxiliary field $\pi$ and ${\sf gh}=-1$ diffeomorphism antighost $\bar \xi$. The Kalb-Ramond 2-form has a reducible gauge symmetry.  In addition to the ghost number ${\sf gh}=0$ gauge potential $B$, there is  the 1-form   ${\sf gh}=1$ ghost $\Lambda$ and  the scalar ${\sf gh}=2$ ghost-for-ghost $\lambda$, and their accompanying   1-form $W, \bar \Lambda$ and 0-form $w, \bar \lambda$ trivial pairs and a final ${\sf gh}=0$ ghost $\eta$.

 The linearised ${\sf gh}=1$ off-shell nilquadratic BRST transformations are given by
\begin{subequations}\label{BRST_grav}
\begin{align} \label{BRST_grav}
 Q_{\mathcal{N}=0} h &= \nabla \xi,\\
 \label{BRST_grav_xi}
Q_{\mathcal{N}=0}\xi&= 0,\\
Q_{\mathcal{N}=0}\bar{\xi}&= \pi,\\
Q_{\mathcal{N}=0} \pi &=0.
\end{align}
\end{subequations}
and $Q_{\mathcal{N}=0}\phi=0$. For the Kalb-Ramond sector see for example \cite{Anastasiou:2018rdx}. 
The physical states are contained in the cohomology of $Q_{\mathcal{N}=0}$.

We focus on the Fierz-Pauli sector. For a given gauge-fixing condition, the linearised $Q_{\mathcal{N}=0}$-invariant BRST action can be written
\be
S^{\rm lin}_{{\rm FP}, \text{BRST}}= \int dx^D\sqrt{-g}  \left(\mathcal{L}_{\rm FP} +Q_{\mathcal{N}=0} \Psi_{h}\right)
\ee
where $\Psi_h$ is the ghost number $-1$ diffeo-gauge-fixing fermion
\be
\Psi_h =     \bar \xi (G[h, \varphi] -  \frac\zeta 2 \pi)
\ee
for gauge-fixing function $G[h, \varphi]$ with Gaussian width $\zeta\in \mathds{R}$. A typical choice of gauge-fixing function is de Donder  gauges $G[A]={\rm div} (h - \frac12 g {\rm tr} h)$, where the trace is taken with respect to the the background metric $g$. This yields 
\be
S^{\rm lin}_{{\rm FP}, \text{BRST}}= \int dx^D\sqrt{-g}  \left(\mathcal{L}_{\rm FP} +\pi \wedge \star G[h, \varphi] - \star \frac\zeta 2 \pi^2 - \bar \xi Q_{\N=0}G[h, \varphi]\right),
\ee
which upon eliminating $\pi$ gives 
\be
S^{\rm lin}_{{\rm FP}, \text{BRST}}=    \int dx^D\sqrt{-g}  \left(\mathcal{L}_{\rm FP}  +\frac{1}{2\zeta}G[h, \varphi] \wedge \star G[h, \varphi]  - \bar \xi Q_{\N=0}G[h, \varphi]\right),
\ee
with 
\be
Q_{\N=0} \bar \xi = \frac1\zeta G[h, \varphi]. 
\ee

\subsection{The Goal: Diffeomorphism BRST from Yang-Mills BRST}\label{brstprod}
We would like to define a product `$\ast$' of fields (i.e.~sections of bundles) and set\footnote{We are ignoring here the bi-adjoint spectator scalar field $\Phi$ \cite{Anastasiou:2014qba}. Since it is a scalar and BRST invariant it can be straightforwardly included. However, since the convolution is not necessarily associative, cf. \autoref{sec:app_associativity}, one must make a choice in defining the product, fixing which field is convoluted with the spectator first. A natural choice is to regard the convolution as a left acting operation.}
\be
h=A\ast \tilde A, \qquad \xi =c\ast \tilde A+A\ast \tilde c.
\ee
Note that since $h\in \mathrm{Sym}(M)$, where $\mathrm{Sym} (M)$ denotes symmetric $(0,2)$-tensor fields on $M$, the product of two 1-forms, $A\ast \tilde{A}$, must output a symmetric tensor. In order to recover the symmetries of linearised gravity, \eqref{BRST_grav} and~\eqref{BRST_grav_xi}, using the symmetries of the two gauge fields, we require
\begin{equation}
Q (A\ast \tilde A) = \nabla \xi \; , \;\;\; Q( c\ast \tilde A+A\ast \tilde c ) = 0 \; ,
\end{equation}
where $Q$ denotes the double-copy transformation that acts either as $Q_{\rm YM}$ or $Q_{\mathcal{N}=0}$, depending on what fields it acts on.

Since $Q A=dc=\nabla c$, it suffices if the product $\ast$ satisfies 
\be\label{eq:goal_derivative_rule}
\nabla f \ast  g = \nabla ( f \ast g ) = f \ast \nabla g , \quad \nabla f \ast  \omega = \nabla ( f \ast \omega ), \quad  \omega \ast \nabla f  = \nabla ( \omega \ast f ) \, ,
\ee
for any scalars $f,g$, and any 1-form $\omega$. Assuming this we find
\begin{align}
Q (A\ast \tilde A) & = QA \ast \tilde{A} + A \ast Q\tilde{A} \nonumber
\\
& = \nabla c \ast \tilde{A} + A \ast \nabla\tilde{c}
\nonumber
\\
& = \nabla \left(  c \ast \tilde{A} + A \ast \tilde{c} \right)
\nonumber
\\
& = \nabla \xi \;\;\; ,
\end{align}
and 
\begin{align}
Q ( c\ast \tilde A+A\ast \tilde c ) & = - c\ast Q \tilde{A} + QA \ast \tilde{c} \nonumber
\\
& = -  c \ast \nabla\tilde{c} + \nabla c \ast \tilde{c}
\nonumber
\\
& = \nabla \left(  - c \ast \tilde{c} + c \ast \tilde{c} \right)
\nonumber
\\
& = 0 \;\;\; ,
\end{align}
as desired. Note, in the last derivation we have used $Qc=Q\tilde{c}=0$ and the anti-commutativity of $Q$ and the ghost fields.

One might also ask for the derivative rule
\begin{equation}
\nabla ( f \ast \omega ) = f \ast \nabla\omega \;\;\; ,
\end{equation}
for any scalar $f$ and 1-form $\omega$. Such a rule would enable the recovery of the dilaton transformation \cite{Borsten:2019prq}. Below we will see that this rule is more complicated than those in~\eqref{eq:goal_derivative_rule}, and hence we leave it as an open question for future investigations. 

We also note that for the Kalb-Ramond 2-form we have 
\be
Q_{\N=0}B=d\Lambda,
\ee
where $d$ is the exterior derivative. The requirement analogous to the function and 1-form case of~\eqref{eq:goal_derivative_rule} is then
\be
d f \ast  w = d ( f \ast w ), \quad  w \ast d f  = d ( w \ast f ).
\ee
Surprisingly, this also proves to be less straightforward and is left for future work. 

Equation~\eqref{eq:goal_derivative_rule} can be seen as the analog of the familiar Leibniz failure property of convolutions on pseudo-Euclidean spaces $\partial (X* Y) = \partial X* Y = X*\partial Y$, but for the symmetrized covariant derivative $\nabla$ on  $M$. \eqref{eq:goal_derivative_rule} is the main goal, and the rest of the paper will be dedicated to constructing a definition of $\ast$ satisfying~\eqref{eq:goal_derivative_rule} on a wide class of spacetime geometries.

\subsection{From Spacetime to Space}\label{spatial}

Specifically, we shall be concerned with $D=d+1$ ultrastatic spacetimes $\hat{M}=\mathds{R}\times M$, where the Riemannian geometry on $M$ is held fixed over time. This essentially reduces the problem to defining a spatial product on the $d$-dimensional Riemannian space $M$. Hence we shall be concerned only with $d$-dimensional objects, and for notational convenience all $(d+1)$-dimensional objects will be hatted henceforth (this is the first and last time we shall use this notation, for reasons we make clear momentarily). Similarly, we will hat the product on $\hat{M}$, $\ast\rightarrow\hat{\ast}$, and the symmetrised covariant derivative $\nabla \rightarrow \hat{\nabla}$. Any occurrence of the un-hatted `$\ast$' will denote an, as of yet, undefined product over the spatial manifold $M$, and any occurrence of the un-hatted `$\nabla$' will refer to the symmetrised covariant derivative for the Levi-Civita connection on $M$.

Given the ultrastatic nature of $\hat{M}$ we can make the convenient time-space split of any 1-form $\hat{\omega}\in \Omega^1(\hat{M})$, i.e. $\hat{\omega} = \omega_0 \, dt + \omega_1$. Here $\omega_0$ is a function over $\hat{M}$, and as such can vary in space as well as time. $\omega_1$ is spatially-directed 1-form, i.e. $\omega_1 (\partial_t )=0$, which again can vary in space and time. In this way we can equivalently think of $\hat{\omega}$ as the pair $(\omega_0[\cdot] , \omega_1[\cdot] ) \in \Omega^0_{\mathbb{R}}(M) \oplus \Omega^1_{\mathbb{R}}(M)$, where $\omega_0[\cdot]\, : \, \mathbb{R}\rightarrow \Omega^0(M)$ denotes a 1-parameter family of functions on $M$, i.e. $\omega_0[t]\in \Omega^0(M)$ for any $t\in \mathbb{R}$, and $\omega_1[\cdot]\, : \, \mathbb{R}\rightarrow \Omega^1(M)$ denotes a 1-parameter family of 1-forms on $M$, i.e. $\omega_1[t]\in \Omega^1(M)$ for any $t\in \mathbb{R}$. Here $\Omega^p_{\mathbb{R}}(M)$ denotes the space of 1-parameter families of $p$-forms. We will write $\omega_i[\cdot]$ for the $i=0,1$ component of this pair. Note that one can go backwards and construct $\hat{\omega}$ given such a pair, and hence this mapping between $\Omega^1(\hat{M})$ and $\Omega^0_{\mathbb{R}}(M) \oplus \Omega^1_{\mathbb{R}}(M)$ is a bijection. 

We can connstruct similar bijections for all $(0,p)$-tensor fields. For functions is it trivial. Given any $\hat{f}\in \Omega^0(\hat{M})$, the corresponding 1-parameter family of functions on $M$, denoted $f[\cdot] \in \Omega^0_{\mathbb{R}}(M)$, is simply $f[t](x) = \hat{f}(t,x)$ for any $t\in \mathbb{R}$ and $x\in M$. The other relevant case for our purposes is the bijection for $\mathrm{Sym} (\hat{M})$. For any $\hat{\alpha}\in \mathrm{Sym} (\hat{M})$, the time-space split can be written in the form 
\begin{equation}
\hat{\alpha} = \alpha_{0} \, dt \vee dt + 2\, \alpha_{1} \vee dt + \alpha_{2} \;\;\; ,
\end{equation}
where $\vee$ denotes the symmetric tensor product, i.e. $\alpha \vee \beta := \alpha\otimes \beta + \beta \otimes \alpha$. Here $\alpha_{0} \in \Omega^0(\hat{M})$ is a function on $\hat{M}$, $\alpha_{1} \in\Omega^1(\hat{M})$ is a spatially-directed 1-form on $\hat{M}$, i.e. $\alpha_{1} (\partial_t ) = 0$, and $\alpha_{2} \in \mathrm{Sym}(\hat{M})$ is a spatially-directed symmetric $(0,2)$-tensor on $\hat{M}$, i.e. $\alpha_{2} (\partial_t , X) = \alpha_{2} (X , \partial_t) = 0$ for any vector $X\in T\hat{M}$.  It is then clear that $\hat{\alpha}$ is equivalent to the triple $( \alpha_{0}[\cdot] , \alpha_{1}[\cdot],\alpha_{2}[\cdot] ) \in \Omega^0_{\mathbb{R}}(M) \oplus \Omega^1_{\mathbb{R}}(M)\oplus \mathrm{Sym}_{\mathbb{R}}(M)$, where $\mathrm{Sym}_{\mathbb{R}}(M)$ denotes the space of 1-parameter families of symmetric $(0,2)$-tensor fields on $M$. We similarly write $\alpha_{i}[\cdot]$, ($i=0,1,2$) for the corresponding component of the triple.

We are now ready to define the spacetime product $\hat{\ast}$ in terms of some, as of yet, undefined spatial product $\ast$. For a pair of functions $\hat{f},\hat{g}\in\Omega^0(\hat{M})$ we first find the corresponding 1-parameter families $f[\cdot],g[\cdot]\in\Omega^0_{\mathbb{R}}(M)$. The spacetime product, $\hat{f}\, \hat{\ast}\, \hat{g}\in\Omega^0(\hat{M})$, can then be defined through a specification of the corresponding 1-parameter family $(\hat{f}\, \hat{\ast}\,\hat{g}) [\cdot]\in\Omega^0_{\mathbb{R}}(M)$, which, for any $t\in\mathbb{R}$, we define as
\begin{equation}\label{eq:spacetime_product_def_functions}
(\hat{f}\, \hat{\ast}\, \hat{g}) [t] := \int_{-\infty}^{\infty}dt' \, f[t']\ast g[t-t'] \;\;\; ,
\end{equation}
where, by assumption, $\ast$ is defined between functions on $M$, and hence the integrand on the RHS is well-defined. Given the bijection between $\Omega^0(\hat{M})$ and $\Omega^0_{\mathbb{R}}(M)$, this specification of the 1-parameter family uniquely defines the product $\hat{f}\, \hat{\ast}\, \hat{g}\in\Omega^0(\hat{M})$. Modulo the definition of $\ast$, the above equation should be familiar to the reader as the usual Euclidean convolution of functions. 

Given two 1-forms $\hat{\alpha},\hat{\beta}\in\Omega^1(\hat{M})$, we wish to go further and define the spacetime products $\hat{f}\, \hat{\ast}\, \hat{\alpha}\in\Omega^1(\hat{M})$ and $\hat{\alpha}\, \hat{\ast}\, \hat{\beta}\in\mathrm{Sym}(\hat{M})$. Both can be uniquely defined by specifying the corresponding tuples; a pair $\Omega^0_{\mathbb{R}}(M) \oplus \Omega^1_{\mathbb{R}}(M)$ for the former and a triple $\Omega^0_{\mathbb{R}}(M) \oplus \Omega^1_{\mathbb{R}}(M)\oplus \mathrm{Sym}_{\mathbb{R}}(M)$ for the latter. We first compute the components $\alpha_i[\cdot]$ and $\beta_j[\cdot]$, where $i,j=0,1$. We then define the components of the respective tuples, for all $t\in\mathbb{R}$, as
\begin{align}\label{eq:spacetime_product_def_1_forms}
(\hat{f}\, \hat{\ast}\, \hat{\alpha})_i [t] & := \int_{-\infty}^{\infty}dt'\, f[t']\, \ast\, \alpha_i[t-t'] \;\;\; ,
\\
(\hat{\alpha}\, \hat{\ast}\, \hat{\beta})_{i+j} [t] & := \int_{-\infty}^{\infty}dt'\, \alpha_{( i}[t']\, \ast\, \beta_{j )}[t-t'] \;\;\; ,
\end{align}
where the brackets denote the usual symmetrised sum: $T_{(ab)}=\frac{1}{2}(T_{ab}+T_{ba})$. Note that the $\ast$ products in the integrands on the RHS's are between tensor fields on $M$, and hence are well-defined by assumption. One can of course extend the above definition to higher rank tensors, but the expressions becomes more complicated and will not be needed for our purposes. 

Given the above definition of $\hat{\ast}$, we can now show that it satisfies~\eqref{eq:goal_derivative_rule} for $\hat{\nabla}$, assuming the spatial product $\ast$ satisfies~\eqref{eq:goal_derivative_rule} for the spatial symmetrised covariant derivative $\nabla$. To see this we first need to know how $\hat{\nabla}$ acts on a function $\hat{f}\in\Omega^0(\hat{M})$ and a 1-form $\hat{\omega}\in\Omega^1(\hat{M})$ in terms of the components $f[\cdot]\in\Omega^0_{\mathbb{R}}(M)$ and $\omega_i[\cdot]\in\Omega^i_{\mathbb{R}}(M)$ ($i=0,1$).

Recall that $\hat{\nabla}\hat{f} = \hat{d}\hat{f}$, where $\hat{d}$ is the exterior derivative on $\hat{M}$. As $\hat{d}\hat{f}$ is a 1-form on $\hat{M}$, it can be specified by a pair $\Omega^0_{\mathbb{R}}(M)\oplus \Omega^1_{\mathbb{R}}(M)$. One can verify that $(\hat{d}\hat{f})_0 [t] = \partial_t ( f[t])$ and $(\hat{d}\hat{f})_1 [t] = d ( f[t])$, for all $t\in\mathbb{R}$, gives the correct specification, where $d$ is the exterior derivative on $M$. Note that we may write $(\hat{d}\hat{f})_1 [t] = \nabla ( f[t])$, as $\nabla \equiv d$ on $\Omega^0(M)$. If we define the pair of operators $D_0 := \partial_t$ and $D_1 := \nabla$, we can write the components of $\hat{\nabla}\hat{f}$ succinctly as
\begin{equation}\label{eq:nabla_hat_on_functions}
(\hat{\nabla}\hat{f})_i[t] = D_a(f[t]) \;\;\; .
\end{equation}

Since $\hat{\nabla}\hat{\omega}$ is a symmetric $(0,2)$-tensor field, it is uniquely specified by a triple $\Omega^0_{\mathbb{R}}(M) \oplus \Omega^1_{\mathbb{R}}(M)\oplus \mathrm{Sym}_{\mathbb{R}}(M)$, where the components are $(\hat{\nabla}\hat{\omega})_k[\cdot]$ for $k=0,1,2$. One can verify that
\begin{equation}\label{eq:nabla_hat_on_1_forms}
(\hat{\nabla}\hat{\omega})_{i+j}[t] = D_{( i} (\omega_{j )}[t]) \;\;\; ,
\end{equation}
for $t\in\mathbb{R}$ and $i,j=0,1$, gives the correct specification of the components.

\eqref{eq:goal_derivative_rule} can now be verified by direct computation. For two functions we have
\begin{align}
\left(\hat{\nabla}( \hat{f}\,\hat{\ast}\,\hat{g})  \right)_i[t] & = D_i\left((\hat{f}\,\hat{\ast}\,\hat{g})[t]\right)
\nonumber
\\
& = D_i\left( \int_{-\infty}^{\infty}dt'\, f[t']\,\ast\,g[t-t']\right) \;\;\; .
\end{align}
We then note that the $D_i$ can be moved onto either argument of the $\ast$ product in the integrand. For $i=1$ ($D_1 = \nabla$) this is true by linearity of the integral and by our assumption that $\ast$ satisfies~\eqref{eq:goal_derivative_rule} for $\nabla$. For $i=0$ ($D_0 = \partial_t$) it is obviously true for $g[t-t']$, as it is the only part of the integrand that explicitly depends on $t$, and we have assumed bilinearity of the $\ast$ product. To show that it can be moved onto $f[t']$ one can run the same argument again following the change of integration variables, $t'\rightarrow u=t-t'$: 
\begin{equation}
\int_{-\infty}^{\infty}dt'\, f[t']\,\ast\,g[t-t'] = \int_{-\infty}^{\infty}du\, f[t-u]\,\ast\,g[u] \;\;\; .
\end{equation}
We can now write
\begin{equation}
\left(\hat{\nabla}( \hat{f}\,\hat{\ast}\,\hat{g})  \right)_i[t] = \int_{-\infty}^{\infty}dt'\, D_i(f[t'])\,\ast\,g[t-t'] = \int_{-\infty}^{\infty}dt'\, f[t']\,\ast\,D_i(g[t-t']) \;\;\; ,
\end{equation}
where any time derivative (for $i=0$) on the RHS should be understood as a derivative with respect to the time variable of the respective function, e.g. $D_0(f[t']) = \partial_{t'}(f[t'])$. From the definition~\eqref{eq:spacetime_product_def_1_forms}, and~\eqref{eq:nabla_hat_on_functions} we then have
\begin{equation}
\left(\hat{\nabla}( \hat{f}\,\hat{\ast}\,\hat{g})  \right)_i[t] = \left( (\hat{\nabla} \hat{f})\,\hat{\ast}\,\hat{g}  \right)_i[t] = \left( \hat{f}\,\hat{\ast}\,(\hat{\nabla}\hat{g})  \right)_i[t] \;\;\; ,
\end{equation}
which verifies the desired derivative rule for a product of functions. The $\hat{\nabla}$ derivative rule for the $\hat{\ast}$ product of a function and a 1-form on $\hat{M}$ can be verified in a similar manner. Again, one only requires $\ast$ to be bilinear and to satisfy the derivative rule for $\nabla$. The question, then, is whether such a $\ast$ product on $M$ can be constructed. Explicitly, we want~\eqref{eq:goal_derivative_rule} for functions and 1-formsm but on $M$ instead of $\hat{M}$. From~\cite{Anastasiou:2014qba, Anastasiou:2018rdx, Borsten:2019prq}, the expectation is that $\ast$ should correspond to some kind of convolution. On a generic spatial manifold, $M$, this is a difficult question, but on compact homogeneous spaces we are able to give a well defined convolution satisfying \eqref{eq:goal_derivative_rule}. This is what we shall develop in the remaining sections. 

\section{Lie Groups}\label{group}

\subsection{Preliminaries}\label{groupprelim}

\subsubsection{Definitions}

Consider a compact Lie group $G$ and the associated Lie algebra $\mathfrak{g} \cong T_e G$. For any $g\in G$ we denote the corresponding right and left action as $R_g$ and $L_g$ respectively, i.e. $R_g g' = g' g$ and $L_g g' = gg'$. Note that $R_g$ and $L_g$ are diffeomorphisms of $G$.

For any function $f\in \Omega^0(G)$, the associated pull-backs, $R_g^* f, L_g^* f \in \Omega^0(G)$, are defined as
\begin{align}
(R_g^* f )(g') & := f(R_g g') = f(g' g) 
\\
(L_g^* f )(g') & := f(L_g g') = f(g g')  \;\;\; .
\end{align}
For any vector, $X_g \in T_g G$, we define the push-forward, ${R_{g'}}_* \, : \, T_g G \rightarrow T_{gg'}G$, as that which satisfies
\begin{equation}
({R_{g'} }_* X_g )_{gg'}(f) := X_g ( R_{g'}^* f )
\end{equation}
for any function $f\in \Omega^0(G)$. Similarly, for ${L_{g'}}_* \, : \, T_g G \rightarrow T_{g'g}G$, we have
\begin{equation}
({L_{g'} }_* X_g )_{g'g}(f) := X_g ( L_{g'}^* f ) \;\;\; .
\end{equation}
For any 1-form $\omega \in \Omega^1(G)$ we define the pull-back associated with $R_g$ as
\begin{equation}
(R_g^* \omega )(X) := \omega( {R_g}_* X ) \;\;\; ,
\end{equation}
for any $X\in TG$, and similarly for the pull-back $L_g^*$.

We can now define the subspace of vector fields, $R G \subset TG$, that are \emph{right-invariant}. For any $X\in RG$ we have
\begin{equation}
({R_{g'}}_*X_g )_{g g'} = X_{gg'} \;\;\; ,
\end{equation}
for any $g,g'\in G$. Similarly, we define $LG \subset TG$ as the subspace of \emph{left-invariant} vector fields which satisfy
\begin{equation}
({L_{g'}}_*X_g )_{g' g} = X_{g'g} \;\;\; .
\end{equation}
Using the pull-back we can similarly define right/left-invariant $(0,p)$-tensor fields. Given that $G$ is compact, there is a unique (up to a constant rescaling) bi-invariant (right- and left-invariant) metric on $G$, cf. for example \cite{Helgason:2001dif}.

To any $x\in \mathfrak{g}\cong T_e G$ we can associate a left-invariant vector field $X \in LG$ by pushing $x$ forward from $e$ to any $g\in G$ with ${L_g}_*$. That is, $X_g = {L_g}_* x$. This sets up an isomorphism between $\mathfrak{g}$ and $LG$ (and similarly for $RG$). 

Recall that for any $X, Y \in TG$, and any push-forward $f_*$, we have $f_* [X,Y] = [ f_* X , f_* Y]$. This implies that for any two vector fields $X, Y \in LG$, the Lie bracket, $[X , Y ]$, is also left-invariant. The Lie bracket, and the isomorphism between $\mathfrak{g}$ and $LG$, then defines the usual Lie bracket on $\mathfrak{g}$.

Consider some orthonormal (with respect to the bi-invariant metric) basis of the Lie algebra, $\lbrace e_a \rbrace$ ($a=1,..., \text{dim}(G)$). In this basis the structure constants, $f_{ab}^{\;\; c}$, satisfy
\begin{equation}\label{eq:structure_constants}
[e_a , e_b] = f_{ab}^{\;\; c} e_c \;\;\; .
\end{equation}
For any $g\in G$ we have the map $\text{Ad}_g \, : \, \mathfrak{g} \rightarrow \mathfrak{g}$, given by $\text{Ad}_g := {L_g}_*{R_{g^{-1}}}_*$. We can then decompose the vector $\text{Ad}_g e_a$ in terms of the basis $\lbrace e_a \rbrace$ to get
\begin{equation}
\text{Ad}_g e_a = (\text{Ad}_g )_a^{\; b} \; e_b \; ,
\end{equation}
where the components of this decomposition, $(\text{Ad}_g )_a^{\; b}$, are the adjoint representation of $G$ in this basis.

Given the basis $\lbrace e_a \rbrace$ we can form the corresponding orthonormal right- and left-invariant bases over $G$, denoted by $r_a$ and $l_a$ respectively. At any $g\in G$ these basis vectors are given by
\begin{subequations}\label{basis}
\begin{align}
(r_a )_g & := {R_g}_* e_a 
\\
(l_a )_g & := {L_g}_* e_a  \;\;\; .
\end{align}
\end{subequations}
Given some coordinates, $x^{\mu}$, the basis vector field $l_a$ can be decomposed as $l_a(x) = l_a^{\mu}(x) \partial_{\mu} $, and similarly for $r_a$. In this way, $l_a$ and $r_a$ can be thought of as vielbeins on $G$, with respective components $l_a^{\mu}(x)$ and $r_a^{\mu}(x)$. 

One can then verify that the left and right bases are related to each other via
\begin{equation}\label{eq:relation_between_left_and_right_bases}
(l_a )_g = (\text{Ad}_g )_a^{\; b} \; (r_b)_g \; ,
\end{equation}
when evaluated at some point $g\in G$. The bi-invariance of the metric ensures that these bases remain orthonormal, and the fact that $LG$ is closed under the Lie bracket implies that
\begin{equation}\label{eq:RL_structure_constants}
[l_a , l_b ] = f_{ab}^{\;\; c} l_c \;\;\; , \;\;\; [r_a , r_b ] = f_{ab}^{\;\; c} r_c \;\;\; .
\end{equation}

We can use these left- and right-invariant bases to decompose any vector field $X \in TG$:
\begin{equation}
X = X^a_R \, r_a = X^a_L \, l_a \;\;\; ,
\end{equation}
where, for each index $a$, the components $ X^a_R $ and $ X^a_L$ are functions on $G$. We denote by $r^a$ and $l^a$ the dual basis 1-forms. With tensor products one can then expand any tensor field using $r^a$, $l^a$, $r_a$, and $l_a$. 

Given a vector at a point $g\in G$, say $X_g \in T_g G$, we can generate the corresponding left/right-invariant vector field by using the push-forward to move it to every other point $g'\in G$. We define  maps,
\begin{align}
\mathfrak{l}: G\times TG & \rightarrow  LG;\qquad (g, X)\mapsto  \mathfrak{l}_g (X) \nonumber
\\
\mathfrak{r}: G\times TG & \rightarrow  RG;\qquad (g, X)\mapsto  \mathfrak{r}_g (X) \;\;\; ,
\end{align}
that generate left- and right-invariant vector fields out of vectors at $g$~\footnote{The map $\mathfrak{l}_g$ would usually be referred to as the Lie algebra valued \emph{Solder} 1-form, $\theta_g \, : \, T_g G \rightarrow LG$ (here we have identified the Lie algebra with the left-invariant vector fields).}. Explicitly, if we evaluate the fields $\mathfrak{l}_g (X )$ and $\mathfrak{r}_g (X)$ at any point $g'\in G$ we have
\begin{align}\label{lrvect}
\mathfrak{l}_g (X )\big|_{g'} & := ({L_{g'g^{-1}}}_* X_g )_{g'} \nonumber
\\
\mathfrak{r}_g (X )\big|_{g'} & := ({R_{g^{-1} g'}}_* X_g )_{g'} \;\;\; .
\end{align}
The right- and left-invariant basis fields clearly satisfy $\mathfrak{r}_g (r_a ) = r_a$ and $\mathfrak{l}_g (l_a ) = l_a$.

For a 1-form $\omega\in \Omega^1(G)$, we define the maps $\mathfrak{l}_g (\cdot )$ and $\mathfrak{r}_g (\cdot )$ in an analogous way, but with the pull-back instead:
\begin{align}\label{eq:definition_mathfrak_lr_maps}
\mathfrak{l}_g(\omega)\big|_{g'} & = ( L_{gg'^{-1}}^* \, \omega )\big|_{g'} \nonumber
\\
\mathfrak{r}_g(\omega)\big|_{g'} & = ( R_{g'^{-1}g}^* \, \omega )\big|_{g'} \;\;\; .
\end{align}
We can then extend the definition of these maps to all $(0,p)$-tensor fields using the pull-back in the obvious way. Note that, for any function $f \in \Omega^0(G)$, $\mathfrak{l}_g (f)$ and $\mathfrak{r}_g (f)$ are constant functions on $G$, and at any point $g'\in G$ they take the value $f(g)$, i.e. $\left( \mathfrak{l}_g (f)\right) (g') = \left( \mathfrak{r}_g (f)\right) (g') = f(g) $. 

Given any diffeomorphism $f \, : \, G \rightarrow G$, and any $h\in\Omega^0(G)$ and $\omega\in \Omega^1(G)$, the pull-back $f^*$ satisfies
\begin{equation}
f^* \, h \omega  = (f^* h ) (f^* \omega)  \;\;\; ,
\end{equation}
where we recall that the pull-back of a function is defined as $f^* h := h \circ f$. If we write $\omega = (\omega_L)_a \, l^a$, we then have that
\begin{align}\label{eq:left_action_1_form_components}
L_g^* \, \omega  & = L_g^* \, ( (\omega_L)_a  \, l^a ) \nonumber
\\
& = (L_g^* \, (\omega_L)_a ) \, (L_g^* \, l^a ) \nonumber
\\
& = (L_g^* \, (\omega_L)_a ) \,  l^a \;\;\; .
\end{align}
In this basis $L_g^*$ simply acts as the pull-back on the component functions. Similarly we have
\begin{equation}\label{eq:right_action_1_form_components}
R_g^* \,  \omega = (R_g^* \, (\omega_R)_a ) \, r^a \;\;\; .
\end{equation}
One can then verify that
\begin{align}\label{eq:RLg_action_1_form_components}
\mathfrak{l}_g(\omega) & = \mathfrak{l}_g((\omega_L)_a ) \, l^a \nonumber
\\
\mathfrak{r}_g(\omega) & = \mathfrak{r}_g( (\omega_R)_a ) \, r^a \;\;\; .
\end{align}

\subsubsection{Derivatives}\label{diffgroup}

We will mostly be concerned with the covariant derivative, $\nabla$, specifically the torsionless connection that is compatible with the bi-invariant metric. Before discussing $\nabla$ we briefly comment on a useful property of the exterior derivative, $d: \Omega^p(M)\rightarrow \Omega^{p+1}(M)$.

 Using the  bases \eqref{basis} we can write
\begin{equation}\label{eq:d_on_functions}
d f = r_a (f) r^a = l_a (f) l^a \;\;\; ,
\end{equation}
for any $f\in \Omega^0(G)$. For higher degree forms the action of $d$ in the left/right-invariant basis is more complicated, but we will not need it in what follows.

Moving on to the covariant derivative, for any pair of left-invariant (or right-invariant) basis vector fields $l_a, l_b \in L G$, $\nabla$ satisfies
\begin{equation}
\nabla_{l_a}l_b = \frac{1}{2}[l_a,l_b] = \frac{1}{2}f_{ab}^{\;\; c}l_c \;\;\; ,
\end{equation}
For any two vector fields $X, Y \in TG$ we can write them in the $l_a$ basis as
\begin{equation}
X = X^a_L \, l_a \;\;\; , \;\;\; Y = Y^a_L \, l_a \;\;\; .
\end{equation}
Using the algebraic properties of $\nabla$, we have
\begin{align}
\nabla_X Y & = \nabla_{ X} ( Y^a_L \, l_a ) \nonumber
\\
& = \nabla_{ X} ( Y^a_L  ) \, l_a + Y^a_L \nabla_{ X }l_a \nonumber
\\
& = X(Y^a_L)\, l_a + \frac{1}{2} X^a_L Y^b_L f_{ab}^{\;\; c}l_c \;\;\; ,
\end{align}
where the covariant derivative of a function $f\in \Omega^0(G)$, with respect to some vector field $X \in T G$, is defined as $\nabla_{X}f := X (f)$. The covariant derivative of some 1-form $\omega\in\Omega^1(G)$ is defined as
\begin{equation}
(\nabla_{X}\omega )(Y) := X( \omega(Y)) - \omega(\nabla_X Y) \;\;\; ,
\end{equation}
for any $X, Y \in TG$. By writing $\omega = (\omega_L)_a \, l^a$ we find
\begin{equation}\label{eq:cov_deriv_1_form_components}
(\nabla_{X}\omega )(Y) = Y_L^a X((\omega_L)_a) - \frac{1}{2}X^a_L Y^b_L f_{ab}^{\;\; c} (\omega_L)_c \;\;\; .
\end{equation}
The definition of the covariant derivative can be extended to any tensor field in the usual way.

For any function $f$ we define the symmetrised covariant derivative $\nabla f \in \Omega^1(G)$ as $\nabla f := df$, and for any $\omega\in \Omega^1(G)$ we define it as
\begin{equation}\label{eq:def_symm_cov_deriv}
(\nabla \omega )(X,Y) := (\nabla_{X}\omega )(Y) + (\nabla_{Y}\omega )(X) \;\;\; ,
\end{equation}
for any $X, Y \in TG$. Note that the symmetrised covariant derivative has the same symbol as the usual covariant derivative. Any ambiguity between the two can be resolved by the fact that the symmetrised covariant derivative does not take a subscript argument, while the usual covariant derivative does.

Using~\eqref{eq:cov_deriv_1_form_components}, and the anti-symmetry of $f_{ab}^{\;\; c}$, we find
\begin{equation}\label{eq:sym_cov_deriv_omega_components}
(\nabla \omega )(X,Y) = X((\omega_L)_a)Y_L^a + Y((\omega_L)_a)X_L^a \;\;\; .
\end{equation}
We also have that
\begin{equation}
\left( d((\omega_L)_a) \vee l^a \right) (X,Y) = X((\omega_L)_a)Y_L^a + Y((\omega_L)_a)X_L^a \;\;\; ,
\end{equation}
where $\vee$ denotes the symmetrised tensor product, i.e. $\alpha \vee \beta := \alpha \otimes \beta + \beta \otimes \alpha$. We can therefore write
\begin{equation}\label{eq:symm_cov_deriv_property}
\nabla \omega = d((\omega_L)_a) \vee l^a = d((\omega_R)_a) \vee r^a \;\;\; ,
\end{equation}
where the last equality follows from the fact that there is no distinction between the properties of left and right fields. In other words, we could re-do the above derivation entirely with right-invariant fields, and we would arrive at the same result. The fact that the symmetrised covariant derivative has this simple form in terms of the left/right-invariant basis will be extremely convenient for our purposes.

\subsection{Convolution}\label{convgroup}

\subsubsection{Definition for Functions}

Consider the group $\mathbb{R}$ with addition. For any $x,x' \in \mathbb{R}$ the group action is defined as $x' \cdot x := x+x'$. The usual convolution of two functions on $\mathbb{R}$ is
\begin{equation}
(f \ast \tilde{f} )(x') = \int_{-\infty}^{\infty} dx f(x) \tilde{f}(x'-x) = \int_{\mathbb{R}} dx f(x) \tilde{f}(x^{-1}\cdot x') \;\;\; ,
\end{equation}
where on the far right we have written it in a form that looks more generalisable to any Lie group, provided the integration measure can be defined. In the above case the group is abelian, so we can equivalently write the argument of $\tilde{f}$ as $x' \cdot x^{-1}$ instead of $x^{-1}\cdot x'$.

For a general group, that may be non-abelian, there are then two obvious choices for the convolution of two functions.
\begin{definition}\label{def:function_conv_group}
For a compact Lie group $G$, we define the \emph{left} and \emph{right} convolution of two functions $f, \tilde{f} \in \Omega^0(G)$ as
\begin{align}\label{eq:function_conv_group}
(f \ast_L \tilde{f})(g') & := \int_G dg \, f(g) \tilde{f}(g^{-1}g')  \; ,
\\
(f \ast_R \tilde{f})(g') & := \int_G dg \, f(g) \tilde{f}(g' g^{-1})  \; ,
\end{align}
where $dg$ is the Haar measure, normalised such that $\mathrm{vol}(G)=1$.
\end{definition}
\noindent For compact groups the Haar measure is invariant under right and left action, and inversion. Note that
\begin{align}
(f \ast_L \tilde{f})(g') & = \int_G dg \, f(g) \tilde{f}(g^{-1}g') \nonumber
\\
& = \int_G d\tilde{g} \, f(g' \tilde{g}^{-1}) \tilde{f}(\tilde{g}) \nonumber
\\
& = (\tilde{f} \ast_R f )(g') \;\;\; ,
\end{align}
where, from  line 1 to line 2, we have changed integration variables from $g$ to $\tilde{g}=g^{-1}g'$, and used the fact that the integration measure is invariant.

We also note that in \cite{Farashahi:2013con} they are able to define the above convolution for more than just smooth functions.

\subsubsection{Definition for Tensor Fields}

Generalising from functions to tensor fields, we \textit{want} to write the right and left convolution as
\begin{align}
(S \otimes_L T)\big|_{g'} & =\int_G dg \, S\big|_{g} \otimes T\big|_{g^{-1} g'}
\\
(S \otimes_R T)\big|_{g'} & =\int_G dg \, S\big|_{g} \otimes T\big|_{ g'g^{-1}}
\end{align}
where $S$ is a $(0,p)$-tensor field, $T$ is a $(0,q)$-tensor field, and the convolution should be a $(0,p+q)$-tensor field. Of course this does not work, as we cannot take tensor products of $S$ and $T$ at different points in the integrand. Since we want the output to be a $(0,p+q)$-tensor at point $g'$, we need some way of `shifting' $S$ from $g$ to $g'$, and $T$ from $g^{-1}g'$, or $g'g^{-1}$, to $g'$. The right and left pull-back, and the maps $\mathfrak{l}_g$ and $\mathfrak{r}_g$, enable us to do this in a way that is consistent with the convolution on functions.

First, we note that we can re-write the right and left convolutions of two functions, $f, \tilde{f}\in \Omega^0(G)$, as
\begin{equation}
\begin{split}
f \ast_L \tilde{f} & = \int_G dg \, \mathfrak{r}_g (f) L_{g^{-1}}^* \tilde{f} 
\\
f \ast_R \tilde{f} & = \int_G dg \, \mathfrak{l}_g (f) R_{g^{-1}}^* \tilde{f} \;\;\; .
\end{split}
\end{equation}
One can verify this by evaluating the above at any point $g'\in G$. Since the left convolution, $\ast_L$, involves the map $\mathfrak{r}_g$ on the first argument, and $L_{g^{-1}}^*$ on the second, we will henceforth refer to is as the \emph{right-left} convolution, or the $RL$-convolution, and denote it by $\ast_{RL}$. Likewise, the right convolution will be referred to as the \emph{left-right}, or $LR$-convolution, and will be written as $\ast_{LR}$. One can also define the $LL$- and $RR$-convolutions, using the respective maps, but since $\mathfrak{r}_g (f) = \mathfrak{l}_g (f)$ for any function $f \in \Omega^0 (G)$, we have $\ast_{LL} = \ast_{RL}$ and $\ast_{RR} = \ast_{LR}$. Notably, this will not be the case for tensor fields.

Given that $\mathfrak{r}_g$, $\mathfrak{l}_g$, and the pull-backs are defined for any $(0,p)$-tensor field, we immediately have our desired definitions for tensor fields:
\begin{definition}\label{def:tensor_convolution_definition_G}
For any $S \in \mathfrak{T}^0_p G$ and any $T \in \mathfrak{T}^0_q G$, the $RL$- and $LR$-convolutions, $\otimes_{RL} , \otimes_{LR}   \, : \, \mathfrak{T}^0_p G \times \mathfrak{T}^0_q G \rightarrow \mathfrak{T}^0_{p+q} G$, are defined as
\begin{align}
S \otimes_{RL} T & := \int_G dg \, \mathfrak{r}_g (S) \otimes L_{g^{-1}}^* T \;\;\; ,
\\
S \otimes_{LR} T & := \int_G dg \, \mathfrak{l}_g (S)\otimes R_{g^{-1}}^* T \;\;\; .
\end{align}
Note, one can define the convolution  on $\mathfrak{T}^p_0 G$ tensors using the left/right push-forwards and the corresponding $\mathfrak{l,r}$ maps for vector fields \eqref{lrvect}. 
\end{definition}
\noindent In general, it is no longer true that $\otimes_{LL} = \otimes_{RL}$ and $\otimes_{RR} = \otimes_{LR}$, as it was for the convolution between functions. Indeed, these identities are only valid for functions and for abelian groups where $l_a = r_a$. Our choice of the `mixed' convolutions, $\otimes_{RL}$ and $\otimes_{LR}$, will become more obvious when we discuss derivatives below.

In what follows we will sometimes write $\wedge_{AB}$ and $\vee_{AB}$ instead of $\otimes_{AB}$, where $A,B=R,L$. In this case one should substitute $\wedge$ and $\vee$ respectively for $\otimes$ in \textbf{Def}.~\ref{def:tensor_convolution_definition_G}.

Using the basis fields $r^a$ and $l^a$ we can decompose two 1-forms $\omega , \tilde{\omega} \in \Omega^1(G)$ as
\begin{align}
\omega & = (\omega_R)_a \, r^a = (\omega_L)_a \, l^a \nonumber
\\
\tilde{\omega} & = (\tilde{\omega}_R)_a \, r^a = (\tilde{\omega}_L)_a \, l^a \;\;\; .
\end{align}
Using~\eqref{eq:left_action_1_form_components},~\eqref{eq:right_action_1_form_components}, and~\eqref{eq:RLg_action_1_form_components}, the $RL$-convolution  between $\omega$ and $\tilde{\omega}$ then simplifies to
\begin{align}
\omega \otimes_{RL} \tilde{\omega} & = \int_G dg \, \mathfrak{r}_g(\omega) \otimes L_{g^{-1}}^* \, \tilde{\omega} \nonumber
\\
& = \int_G dg \, \mathfrak{r}_g((\omega_R)_a) \, r^a \otimes ( L_{g^{-1}}^* \, (\tilde{\omega}_L)_b ) \, l^b \nonumber
\\
& = \left( \int_G dg \, \mathfrak{r}_g((\omega_R)_a) \, ( L_{g^{-1}}^* \, (\tilde{\omega}_L)_b ) \right) \,  r^a \otimes l^b \nonumber
\\
& = \Big( (\omega_R)_a \ast_{RL}  (\tilde{\omega}_L)_b \Big) r^a \otimes l^b \;\;\; ,
\end{align}
and similarly,
\begin{equation}\label{eq:LR_conv_for_1_forms_using_LR_bases}
\omega \otimes_{LR} \tilde{\omega} = \Big( (\omega_L)_a \ast_{LR}  (\tilde{\omega}_R)_b \Big) l^a \otimes r^b \;\;\; .
\end{equation}
From this it is clear that the tensor convolutions we have defined amount to convolutions of the component functions in the right- and left-invariant bases. This generalises to convolutions between any $(0,p)$-tensor fields.

We briefly note the following properties of the $AB$-convolutions:

\noindent\textit{Commutativity} - Writing the convolution in terms of the right- and left-invariant bases it is then clear that
\begin{align}
\omega \vee_{AB} f & = f \vee_{BA} \omega \nonumber
\\
\omega \vee_{AB} \tilde{\omega} & = \tilde{\omega} \vee_{BA} \omega \;\;\; ,
\end{align}
for any $f \in \Omega^0(G)$ and any $\omega , \tilde{\omega}\in \Omega^1(G)$. For any $\alpha \in \Omega^p(G)$, and any $\beta \in \Omega^q(G)$, it is also clear that
\begin{equation}
\alpha \wedge_{AB} \beta = (-1)^{pq} \beta \wedge_{BA} \alpha \;\;\; .
\end{equation}

\noindent\textit{Associativity} - In appendix~\ref{sec:app_associativity} we prove the following associativity property:
\begin{equation}
S \otimes_{AB} ( T \otimes_{AB} U ) = ( S \otimes_{AB} T )\otimes_{AB} U \;\;\;  ,
\end{equation}
for any $S \in \mathfrak{T}^0_p G$, $T \in \mathfrak{T}^0_q G$, and $U \in \mathfrak{T}^0_r G$.

Of particular importance for our purposes is how the convolution acts under derivatives, which we will now discuss.

\subsection{Differentiation}\label{diffgroup}

\subsubsection{Exterior Differentiation}

In this section we will be concerned with the wedge convolution on differential forms and its action under the exterior derivative. First, we show the following
\begin{prop}\label{prop:exterior_deriv_property_for_conv_on_funcs}
For any pair of functions, $f , \tilde{f} \in \Omega^0(G)$, their $AB$-convolution (for $AB=RL,LR$) satisfies
\begin{equation}
d(f \ast_{AB} \tilde{f} ) = df \otimes_{AB} \tilde{f} = f \otimes_{AB} d\tilde{f}.
\end{equation}
\end{prop}
\begin{proof}
To see this consider the $RL$-convolution between two functions, $f \ast_{RL} \tilde{f}$. Using~\eqref{eq:d_on_functions} the exterior derivative of the resulting function can be written as
\begin{equation}
d( f\ast_{RL} \tilde{f} ) = l_a ( f\ast_{RL} \tilde{f} ) \, l^a \;\;\; .
\end{equation}
Consider any left-invariant basis vector $l_a$. For any $g'\in G$ we then have
\begin{align}
l_a ( f\ast_{RL} \tilde{f} )\big|_{g'} & = l_a  \left( \int_G dg \, f(g) \tilde{f}(g^{-1}g') \right)\Big|_{g'} \nonumber
\\
& = \frac{d}{d\lambda}\left( \int_G dg \, f(g) \tilde{f}(g^{-1}\gamma(\lambda)) \right)\Big|_{\lambda=0} \;\;\; ,
\end{align}
where the curve $\gamma(\lambda )$ is a representative of $l_a$ at $g'$, i.e. $\gamma(\lambda )=g' \exp (\lambda e_a )$ where $\exp \, : \,  \mathfrak{g} \rightarrow G$ is the exponential map. Continuing on we have
\begin{align}
l_a ( f\ast_{RL} \tilde{f} )\big|_{g'} & = \frac{d}{d\lambda}\left( \int_G dg \, f(g) \tilde{f}(g^{-1}\gamma(\lambda)) \right)\Big|_{\lambda=0} \nonumber
\\
& =  \int_G dg \, f(g) \frac{d}{d\lambda}\left( \tilde{f}(g^{-1}\gamma(\lambda)) \right)\Big|_{\lambda=0} \nonumber
\\
& =  \int_G dg \, f(g) l_a\left( L_{g^{-1}}^*\tilde{f} \right)\big|_{g'} \nonumber
\\
& =  \int_G dg \, f(g) ({L_{g^{-1}}}_* l_a)_{g'} (\tilde{f}) \nonumber
\\
& = \int_G dg \, f(g) (l_a)_{g^{-1}g'} (\tilde{f}) \;\;\; ,
\end{align}
where the last line follows from the left-invariance of $l_a$. We have just shown that
\begin{equation}
l_a ( f\ast_{RL} \tilde{f} ) = f \ast_{RL} l_a(\tilde{f}) \;\;\; .
\end{equation}
Similarly, one can show that 
\begin{equation}
r_a ( f\ast_{LR} \tilde{f} ) = f \ast_{LR} r_a(\tilde{f}) \;\;\; .
\end{equation}
Together with the commutativity of $\ast_{AB}$ we have
\begin{equation}
l_a ( f\ast_{LR} \tilde{f} ) = l_a ( \tilde{f}\ast_{RL} f ) = \tilde{f}\ast_{RL} l_a(f) = l_a(f) \ast_{LR} \tilde{f} \;\;\; , 
\end{equation}
and similarly
\begin{equation}
r_a ( f\ast_{RL} \tilde{f} ) = r_a(f) \ast_{RL} \tilde{f} \;\;\; .
\end{equation}
We can now show that
\begin{align}
f \otimes_{RL} d\tilde{f} & = f \otimes_{RL} (l_a(\tilde{f}) l^a) \nonumber
\\
& = \left( f \ast_{RL} l_a(\tilde{f}) \right)\, l^a \nonumber
\\
& = l_a\left( f \ast_{RL} \tilde{f} \right)\, l^a \nonumber
\\
& = d(f \ast_{RL} \tilde{f} ) \;\;\; ,
\end{align}
and that
\begin{align}
df \otimes_{RL} \tilde{f} & = (r_a(f) r^a ) \otimes_{RL} \tilde{f} \nonumber
\\
& = \left( r_a(f) \ast_{RL} \tilde{f} \right)\, r^a \nonumber
\\
& = r_a \left( f \ast_{RL} \tilde{f} \right) \, r^a \nonumber
\\
&= d( f \ast_{RL} \tilde{f} ) \;\;\; .
\end{align}
Similarly, one can show that
\begin{equation}
d(f \ast_{LR} \tilde{f} ) = df \otimes_{LR} \tilde{f} = f \otimes_{LR} d\tilde{f} \;\;\; .
\end{equation}
\end{proof}

This property under the exterior derivative is why the mixed convolutions, $\otimes_{RL}$ and $\otimes_{LR}$, are more useful for our purposes. Unfortunately, this sort of derivative property does not continue to higher degree forms. For example, for $\omega \in \Omega^1(G)$ we have
\begin{align}
d ( f \wedge_{RL} \omega ) & = d ( f \otimes_{RL} \omega ) \nonumber
\\
& = d\Big(  (f \ast_{RL} (\omega_L)_a ) \, l^a  \Big) \nonumber
\\
& = d \Big(f \ast_{RL} (\omega_L)_a \Big) \wedge l^a + \Big(f \ast_{RL} (\omega_L)_a \Big) dl^a \;\;\; ,
\end{align}
whereas
\begin{align}
df \wedge_{RL} \omega & = \Big( r_a(f)r^a \Big) \wedge_{RL} \left( (\omega_L)_b l^b \right) \nonumber
\\
& = \Big( r_a(f) \ast_{RL} (\omega_L)_b \Big) \, r^a \wedge l^b \nonumber
\\
& = r_a \Big( f \ast_{RL} (\omega_L)_b \Big)  \, r^a \wedge l^b \nonumber
\\
& = d \Big( f \ast_{RL} (\omega_L)_b \Big) \wedge l^a \;\;\; ,
\end{align}
and so
\begin{equation}
d ( f \wedge_{RL} \omega ) = df \wedge_{RL} \omega + \big(f \ast_{RL} (\omega_L)_a \big) dl^a \;\;\; .
\end{equation}
Similarly, one can show that 
\begin{equation}
d ( \omega \wedge_{RL} f ) = - \omega \wedge_{RL} df + ( (\omega_R)_a \ast_{RL} f  ) dr^a \;\;\; .
\end{equation}
The extra term on the far right in both expressions spoils the distributivity of $d$ that we want for the convolution, and is the reason we cannot use our convolution, as defined, to recover the linearised BRST symmetries of the 2-form field. On the other hand, the symmetrised covariant derivative, $\nabla$, \emph{does} satisfy the desired distributivity property with respect to our convolution.

\subsubsection{Covariant Differentiation}

Here we show that the symmetrised covariant derivative, defined in~\eqref{eq:def_symm_cov_deriv} and also denoted by $\nabla$ (but with no subscript argument), satisfies the following
\begin{prop}\label{prop:cov_deriv_property_on_G}
For any $f, \tilde{f}\in \Omega^0(G)$ and any $\omega\in \Omega^1(G)$,
\begin{equation}\label{eq:cov_deriv_property_on_G}
\nabla ( f \vee_{AB} \tilde{f} ) = \nabla f \vee_{AB} \tilde{f} = f \vee_{AB} \nabla \tilde{f} \;\;\; , \;\;\; \nabla\left( f  \vee_{AB} \omega \right) = \nabla f  \vee_{AB} \omega \;\;\; ,
\end{equation}
where $AB = RL, LR$.
\end{prop}

\begin{proof}
Recall that on functions the symmetric covariant derivative acts as $\nabla f = df$. From \textbf{Prop}.~\ref{prop:exterior_deriv_property_for_conv_on_funcs} we immediately have
\begin{equation}
\nabla ( f \vee_{AB} \tilde{f} ) = \nabla ( f \ast_{AB} \tilde{f} ) = \nabla f \vee_{AB} \tilde{f} = f \vee_{AB} \nabla \tilde{f} \;\;\; ,
\end{equation}
for any two functions $f, \tilde{f} \in\Omega^0(G)$. For any 1-form $\omega\in\Omega^1(G)$ we have
\begin{align}
\nabla f \vee_{RL} \omega & = df \vee_{RL} \omega \nonumber
\\
& = r_a(f)r^a \vee_{RL} (\omega_L)_b l^b \nonumber
\\
& = \Big( r_a(f) \ast_{RL} (\omega_L)_b \Big) \, r^a \vee l^b \nonumber
\\
& = r_a \big( f \ast_{RL} (\omega_L)_b \big) \, r^a \vee l^b \nonumber
\\
& = d\big( f \ast_{RL} (\omega_L)_b \big)\vee l^b \;\;\; .
\end{align}
Recalling~\eqref{eq:symm_cov_deriv_property} we have
\begin{align}
\nabla\left( f \vee_{RL} \omega \right) & = \nabla\big(  f \vee_{RL} ((\omega_L)_a l^a ) \big) \nonumber
\\
& = \nabla\Big( \big( f \ast_{RL} (\omega_L)_a \big) \; l^a \Big) \nonumber
\\
& = d\big( f \ast_{RL} (\omega_L)_a \big) \vee l^a \;\;\; ,
\end{align}
as $f \ast_{RL} (\omega_L)_a $ are the $l^a$ components of $f \vee_{RL} \omega$. Comparing the previous two equations we have our desired derivative property:
\begin{equation}
\nabla\left( f \vee_{RL} \omega \right) = \nabla f \vee_{RL} \omega \;\;\; .
\end{equation}
With a similar calculation one can verify the analogous property for $\vee_{LR}$.
\end{proof}

Not required in~\eqref{eq:goal_derivative_rule}, but true nonetheless, is the following
\begin{prop}\label{prop:cov_deriv_property_2_on_G}
For any $f\in \Omega^0(G)$ and any $\omega\in \Omega^1(G)$,
\begin{equation}\label{eq:cov_deriv_property_2_on_G}
\nabla\left( f  \vee_{AB} \omega \right) =  f  \vee_{AB} \nabla \omega \;\;\; ,
\end{equation}
where $AB = RL, LR$.
\end{prop}

\begin{proof}
Above we saw that
\begin{equation}
\nabla\left( f \vee_{RL} \omega \right) = d\big( f \ast_{RL} (\omega_L)_a \big) \vee l^a \;\;\; ,
\end{equation}
and from~\eqref{eq:symm_cov_deriv_property} we have
\begin{align}
 f \vee_{RL} \nabla \omega  & =  f \vee_{RL} \left( d ((\omega_L)_a ) \vee l^a \right) \nonumber
\\
& = f \vee_{RL} \left( l_b ((\omega_L)_a ) l^b \vee l^a \right) \nonumber
\\
& = \left( f \ast_{RL} l_b ((\omega_L)_a ) \right) l^b \vee l^a \nonumber
\\
& = l_b \left(  f \ast_{RL} (\omega_L)_a  \right) l^b \vee l^a
\nonumber
\\
& = d \left(  f \ast_{RL} (\omega_L)_a  \right) \vee l^a
\;\;\; .
\end{align}
The same can be shown for $\vee_{LR}$.
\end{proof}

\subsection{Analogue of the Convolution Theorem}\label{convthrm}

It is worth pausing to comment on the analogue of the Convolution Theorem, given the above definition of a tensor field convolution on the compact Lie group $G$. The analogue of the Convolution Theorem is well understood for the convolution of two functions on $G$, and so the only generalisation we are providing here is an extension to tensor fields.

Recall that the Convolution Theorem on Euclidean space states that the Fourier coefficients of $f\ast \tilde{f}$ (where $f$ and $\tilde{f}$ are functions) are given by the pointwise product (in Fourier space) of the Fourier coefficients of $f$ and $\tilde{f}$, up to some normalisation factor.

On a general Lie group $G$ the situation is slightly more complicated, even for the convolution of two functions. For any function, $f\in\Omega^0(G)$, the Peter-Weyl Theorem~\cite{peterweyl} tells us that we can decompose it as
\begin{equation}
f(g) = \sum_{[\rho]\in\hat{G}}\sum_{i,j=1}^{d_{\rho}} f(\rho | i,j ) \sqrt{d_{\rho}}\, \phi^{\rho}_{ij}(g) \;\;\; .
\end{equation}
The first sum is over equivalence classes, $[\rho]$, of irreducible unitary representations of $G$, where the equivalence is up to isomorphism. $\hat{G}$ denotes the space of all such equivalence classes. $d_{\rho}$ is the dimension of the representation $\rho$, and $f(\rho | i,j)$ denotes the corresponding `Fourier coefficient' for the representation $\rho$ and indices $i,j=1,..., d_{\rho}$. Lastly, $\phi^{\rho}_{ij}(g)$ denotes the $(i,j)$-th matrix element of the $\rho$ representation of $g\in G$, in some orthonormal basis of the vector space used in the representation, i.e. $\phi^{\rho}_{ij}(g) := \braket{v_i, \rho(g) v_j}$ for some orthonormal basis $\lbrace v_i \rbrace_{i=1,...,d_{\rho}}$. One can then verify that
\begin{align}\label{eq:peter_weyl_identities}
\phi^{\rho}_{ij}(g^{-1}) & = \phi^{\rho}_{ji}(g)^* \nonumber
\\
\phi^{\rho}_{ij}(g h) & = \sum_{k=1}^{d_{\rho}} \phi^{\rho}_{ik}(g)\phi^{\rho}_{kj}(h) \nonumber
\\
\int_G dg \, \phi^{\rho}_{ij}(g)^* \phi^{\rho'}_{i'j'}(g ) & = d_{\rho}^{-1}\delta_{[\rho] [\rho']} \delta_{ii'}\delta_{jj'} \;\;\; ,
\end{align} 
where $\cdot^*$ denotes complex conjugation.

Now consider the convolution of two functions, $f\ast_{AB}\tilde{f}$. The analogue of the Convolution Theorem on Euclidean space is then a relationship between the Fourier coefficients of $f\ast_{AB}\tilde{f}$, i.e. $( f\ast_{AB}\tilde{f} )(\rho | i,j)$, and the Fourier coefficients $f(\rho|i,j)$ and $\tilde{f}(\rho|i,j)$. Using the identities in~\eqref{eq:peter_weyl_identities}, we find that $f\ast_{RL}\tilde{f}$ evaluated at some $g\in G$ simplifies to
\begin{equation}
(f\ast_{RL}\tilde{f})(g) = \sum_{[\rho]\in\hat{G}}\sum_{i,j,k=1}^{d_{\rho}}f(\rho|i,k)\tilde{f}(\rho|k,j)\phi^{\rho}_{ij}(g) \;\;\; ,
\end{equation}
and hence
\begin{equation}
( f\ast_{RL}\tilde{f} )(\rho | i,j) = \frac{1}{\sqrt{d_{\rho}}}\sum_{k=1}^{d_\rho}f(\rho|i,k)\tilde{f}(\rho|k,j) \;\;\; .
\end{equation}
This can be written more succinctly if we let $f(\bm{\rho})$ denote the $d_{\rho}\times d_{\rho}$ `Fourier matrix' with elements $f(\rho|i,j)$. The above equation can then be written as
\begin{equation}\label{eq:fourier_matrix_RL}
(f\ast_{RL}\tilde{f} )(\bm{\rho}) = \frac{1}{\sqrt{d_{\rho}}}f(\bm{\rho}).\tilde{f}(\bm{\rho}) \;\;\; ,
\end{equation}
where `$.$' denotes matrix multiplication. The comparison to the Convolution Theorem is now clear. Specifically, the Fourier matrix of the convolution is given by the pointwise\footnote{Pointwise in the space $\hat{G}$.} matrix product of the Fourier matrices of the input functions. If the group $G$ is abelian, then the representations are all 1-dimensional, and this matrix product reduces to normal multiplication. In this case one recovers the direct analogue of the Convolution Theorem on Euclidean space. Similarly, for the $LR$-convolution one finds
\begin{equation}
(f\ast_{LR}\tilde{f} )(\bm{\rho}) = \frac{1}{\sqrt{d_{\rho}}}\tilde{f}(\bm{\rho}).f(\bm{\rho}) \;\;\; .
\end{equation}

Now that we have established the analogue of the Convolution Theorem for functions on $G$, we turn our attention to the convolution of two tensor fields $S\in \mathfrak{T}^0_p (G)$ and $T\in \mathfrak{T}^0_q (G)$. After expanding $S$ and $T$ in terms of the right and left bases repsectively, we can use~\eqref{eq:LR_conv_for_1_forms_using_LR_bases} to write $S\otimes_{RL}T$ as 
\begin{equation}
S\otimes_{RL}T = \left( (S_R )_{a_1 ... a_p}\ast_{RL}(T_L )_{b_1 ... b_q} \right) r^{a_1}\otimes ... \otimes r^{a_p} \otimes l^{b_1} \otimes ... \otimes l^{b_q} \;\;\; .
\end{equation}
Since the components $(S_R )_{a_1 ... a_p}$ and $(T_L )_{b_1 ... b_q}$ are functions on $G$ (for fixed indices $a_1 ,... ,b_q$), we can decompose them using the Peter-Weyl Theorem and find their corresponding Fourier matrices $(S_R )_{a_1 ... a_p}(\bm{\rho})$ and $(T_L )_{b_1 ... b_q}(\bm{\rho})$. The Fourier matrix associated to the convolution of the components, $ (S_R )_{a_1 ... a_p}\ast_{RL}(T_L )_{b_1 ... b_q} $, is then given by
\begin{equation}\label{eq:S_convolve_RL_T_fourier_matrix}
\left( (S_R )_{a_1 ... a_p}\ast_{RL}(T_L )_{b_1 ... b_q} \right)(\bm{\rho}) = \frac{1}{\sqrt{d_{\rho}}} (S_R )_{a_1 ... a_p}(\bm{\rho}) . (T_L )_{b_1 ... b_q}(\bm{\rho}) \;\;\; ,
\end{equation}
using~\eqref{eq:fourier_matrix_RL} above. This gives us the tensor generalisation of the Convolution Theorem. The first difference to the scalar case is that Fourier matrices, $(S_R )_{a_1 ... a_p}(\bm{\rho})$ and $(T_L )_{b_1 ... b_q}(\bm{\rho})$, are constructed for each component of the tensors $S$ and $T$, when expressed in the relevant right and left bases. Second, the matrix multiplication of $(S_R )_{a_1 ... a_p}(\bm{\rho})$ and $(T_L )_{b_1 ... b_q}(\bm{\rho})$ does not give the Fourier matrix of $S\otimes_{RL}T$ directly. It only gives the the Fourier matrix of the components of $S\otimes_{RL}T$, as expressed in the right and left bases. Specifically, if we expand $S\otimes_{RL}T$ as
\begin{equation}
S\otimes_{RL}T = ( S\otimes_{RL}T )_{a_1 , ... , b_q } r^{a_1}\otimes ... \otimes r^{a_p} \otimes l^{b_1} \otimes ... \otimes l^{b_q} \;\;\; ,
\end{equation}
then the Fourier matrix for the components $( S\otimes_{RL}T )_{a_1 , ... , b_q }$ is given by~\eqref{eq:S_convolve_RL_T_fourier_matrix}.

Similarly, for the convolution $S\otimes_{LR}T$ we can expand it as
\begin{equation}
S\otimes_{LR}T = ( S\otimes_{LR}T )_{a_1 , ... , b_q } l^{a_1}\otimes ... \otimes l^{a_p} \otimes r^{b_1} \otimes ... \otimes r^{b_q} \;\;\; .
\end{equation}
The Fourier matrices for the components $( S\otimes_{LR}T )_{a_1 , ... , b_q }$ are then given by
\begin{equation}
( S\otimes_{LR}T )_{a_1 , ... , b_q }(\bm{\rho}) = \frac{1}{\sqrt{d_{\rho}}} (T_R )_{b_1 ... b_q}(\bm{\rho}) . (S_L )_{a_1 ... a_p}(\bm{\rho})  \;\;\; .
\end{equation}

\section{Homogeneous Spaces}
The convolution of functions on  Riemannian  homogeneous spaces $G/H$ is reasonably well-developed, cf.~for example  \cite{Farashahi:2015} and the references therein. Here we generalise to tensor fields. 
\subsection{Preliminaries}\label{homrev}

\subsubsection{Definitions}

Consider a Riemannian manifold $M$ on which a compact Lie group $G$ acts transitively. Pick a base point $\eta\in M$ and consider the subgroup $H$ for which $h\eta =\eta$ for all $h\in H$. Then $M \cong G / H$, and $G$ can be thought of as a $H$-principal bundle over $M$ with the projection $\pi \, : \, G \rightarrow M$, $ g \rightarrow \pi(g) := g \eta $.

Given some $X \in TG$, the push-forward gives the vector $\pi_* X \in TM$, which acts on a function $f\in \Omega^0(M)$, at the point $g\eta \in M$, as
\begin{equation}
(\pi_* X)_{g\eta}(f) = X_g ( \pi^* f ) = X_g ( f \circ \pi ) \;\;\; .
\end{equation}
For any vector $X_g \in T_g G$ we can find a representative curve $\gamma(\lambda )$, with $\gamma(0)=g$ and such that $\gamma(\lambda )$ is tangent to $X_g$ at $g$. Consider the special case where the vector $X_g$ has a representative curve of the form $\gamma(\lambda) = g h(\lambda )$, where $h(\lambda ) \in H$ and where $g h(0)=g \, e = g$. The action of the vector $\pi_* X_g \in T_{g\eta}M$ on any function $f\in \Omega^0(M)$ is then
\begin{align}
(\pi_* X_g )_{g\eta}(f) & = X_g ( f \circ \pi ) \nonumber
\\
& = \frac{d}{d\lambda}\left[ (f\circ \pi)(g h(\lambda ) ) \right]\Big|_{\lambda=0} \nonumber
\\
& = \frac{d}{d\lambda}\left[ f( g h(\lambda ) \eta ) \right]\Big|_{\lambda=0} \nonumber
\\
& = \frac{d}{d\lambda}\left[ f( g \eta ) \right]\Big|_{\lambda=0} \nonumber
\\
& = 0 \;\;\; ,
\end{align}
and hence $(\pi_* X_g )_{g\eta} = 0$. The vector $X_g$ is then in the kernel of the push-forward $\pi_*$. Furthermore, if $X_g \in \text{ker}(\pi_*)$ then ${L_{g'}}_* \, X_g$ and ${R_{h}}_* \, X_g$ are also in the kernel for any $g'\in G$ and any $h\in H$. This means that the kernel of $\pi_*$ is left-invariant \emph{and} $R_H$-invariant. We call the kernel of $\pi_*$ the \emph{vertical} subspace of $TG$, and denote it by $\mathcal{V}TG$. 

Using the bi-invariant metric on $G$ we can find the subspace of $TG$ that is orthogonal to $\mathcal{V}TG$. This is the \emph{horizontal} subspace, denoted by $\mathcal{H}TG$. From the bi-invariance of the metric, and the left- and $R_H$-invariance of $\mathcal{V}TG$, one can see that $\mathcal{H}TG$ is also left- and $R_H$-invariant. We can now write $TG = \mathcal{H}TG \oplus \mathcal{V}TG$. For any $X \in TG$ we have
\begin{equation}
X = \mathcal{H}X + \mathcal{V}X  \;\;\; ,
\end{equation}
where $\mathcal{H}$ ($\mathcal{V}$) is the horizontal (vertical) projector.

Recall that $\mathfrak{g}\cong T_e G$. We now have the orthogonal decomposition $\mathfrak{g} = \mathfrak{m} \oplus \mathfrak{h}$, where $\mathfrak{m} = \mathcal{H}T_e G$ and $\mathfrak{h} = \mathcal{V}T_e G$. One can also show that $\mathfrak{h}$ is the Lie algebra of the subgroup $H$.

Let us pick our basis, $\lbrace e_a \rbrace$, of $\mathfrak{g}$, such that $\lbrace e_i  \rbrace$ is a basis for $\mathfrak{m}$, where $i= 1 , ... , \text{dim}(\mathfrak{m})$, and such that $e_A$ is a basis for $\mathfrak{h}$, where $A = \text{dim}(\mathfrak{m})+1 , ... , \text{dim}(G) $. As $\mathcal{H}TG$ and $\mathcal{V}TG$ are left-invariant, the left-invariant basis $\lbrace l_a \rbrace$ also splits in a similar way, i.e. $\lbrace l_i \rbrace$ is a basis of $\mathcal{H}TG$, and $\lbrace l_A \rbrace$ is a basis of $\mathcal{V}TG$. Any vector field $X \in TG$ can then be written as
\begin{align}
X & = \mathcal{H}X + \mathcal{V}X \nonumber
\\
& = X_L^i \, l_i + X_L^A \, l_A \;\;\; .
\end{align}
For any $\omega \in \Omega^1(G)$ we can decompose it using the dual basis:
\begin{align}\label{eq:hor_vert_decomp_1_form}
\omega & = \mathcal{H}\omega + \mathcal{V}\omega \nonumber
\\
& = (\omega_L)_i \, l^i +(\omega_L)_A \, l^A \;\;\; ,
\end{align}
and similarly for any $(p,q)$-tensor field.

We note that
\begin{equation}\label{eq:RH_commutes_with_H}
{R_h}_* \circ \mathcal{H} = \mathcal{H} \circ {R_h}_* \;\;\; ,
\end{equation}
for any $h\in H$. To see this, we first note that the vertical subspace is $R_H$-invariant, which means that ${R_h}_* \, l_A$ (for any $h\in H$) can be expanded in terms of the $l_A$ basis only. From the right-invariance of the metric we know that $\lbrace {R_h}_* \, l_A \rbrace$ is an orthonormal set, and since each vector ${R_h}_* \, l_A$ lies in $\mathcal{V}TG$, we know that $\lbrace {R_h}_* \, l_A \rbrace$ is an orthonormal basis of $\mathcal{V}TG$. From the right-invariance of the metric we also know that ${R_h}_* \, l_i$ is orthogonal to each ${R_h}_* \, l_A$, and hence ${R_h}_* \, l_i \in \mathcal{H}TG$. This means that ${R_h}_* \, l_i$ can be expanded in terms of the $l_i$ basis only. As the set  $\lbrace {R_h}_* \, l_i \rbrace$ is orthonormal (from the right-invariance of the metric) we then have that $\lbrace {R_h}_* \, l_i \rbrace$ is a basis of $\mathcal{H}TG$. The preceeding argument shows that ${R_h}_* $ acts separately on the bases $\lbrace l_A \rbrace$ and $\lbrace l_i \rbrace$, and rotates them amongst themselves. This means that ${R_h}_*$ and $\mathcal{H}$ commute as desired.

For any function $f \in \Omega^0(M)$, recall that the pull-back $\pi^* f \in \Omega^0(G)$ is defined as $\pi^* f := f \circ \pi$. It is then clear that $\pi^* f$ is $R_H$-invariant. That is, $R_h^* \pi^* f = \pi^* f$, which can be verified by acting on any $g\in G$. For any $\omega \in \Omega^1(M)$ the pull-back $\pi^* \omega\in\Omega^1(G)$ is defined via its action on any $X\in TG$:
\begin{equation}
(\pi^* \omega ) (X)\big|_{g} := \omega( \pi_* X )\big|_{g\eta} \;\;\; ,
\end{equation}
for any $g\in G$.

\begin{definition}
Given some $X \in TM$, the \textbf{horizontal lift}, $\mathfrak{L}(X) \in \mathcal{H}TG$, is the unique horizontal vector field which satisfies
\begin{equation}
( \pi_* \mathfrak{L}(X)_g )_{g\eta} = X_{g\eta} \;\;\; ,
\end{equation}
for any $g\in G$.
\end{definition}
\noindent More simply, we can write $\pi_* \mathfrak{L}(X) = X$. Note that this condition implies that $\mathfrak{L}(X)$ is $R_H$-invariant, as for any $Y_g \in T_g G$ we have
\begin{equation}
( \pi_* ( {R_h}_* \, Y_g )_{gh} )_{g\eta} = ( \pi_*  Y_g )_{g\eta} \;\;\; ,
\end{equation}
for any $h\in H$. 

We note some basic identities that will be useful later. For any $X\in TM$, $Y_g \in T_g G$, $f\in \Omega^0(M)$, and $\omega \in \Omega^1(M)$,
\begin{align}
\pi_* \mathfrak{L}(X) & = X  \; , \label{eq:lift_identities_1}
\\
\mathfrak{L}( (\pi_* Y_g )_{g\eta} )_g & = \mathcal{H} Y_g  \; ,  \label{eq:lift_identities_2}
\\
(\mathfrak{L}(X))_{g}(\pi^* f) & = X_{g\eta}(f) \; ,  \label{eq:lift_identities_3}
\\
(\pi^* \omega)(\mathfrak{L}(X))\big|_g & = \omega(X)\big|_{g\eta} \; . \label{eq:lift_identities_4}
\end{align}
We will also need the following map that averages along the fiber $H$:
\begin{definition}
For any $(0,p)$-tensor field (including functions) we define the map, $A_H \, : \mathfrak{T}^0_p G \rightarrow \mathfrak{T}^0_p G$, that averages along the fiber $H$ as
\begin{equation}
A_H := \frac{1}{\mathrm{vol}(H)}\int_H dh \, R_h^* \;\;\; ,
\end{equation}
where $dh$ is the Haar measure for the compact subgroup $H$.
\end{definition}
\noindent Clearly, for any $R_H$-invariant $T \in \mathfrak{T}^0_p G$ we have $A_H T = T$. Note that, for any $h'\in H$, we have
\begin{align}
R_{h'}^* A_H   = A_H \;\;\; ,
\end{align}
where we have used the invariance of $dh$ under right action by $H$. For any $T \in \mathfrak{T}^0_p G$, we can then see that $A_H T$ is $R_H$-invariant. For any $f\in \Omega^0(G)$, the function $A_H f$ is constant along the fibers of $H$, and its value along any fiber is the average of $f$ along that fiber.

Given~\eqref{eq:RH_commutes_with_H}, we have
\begin{equation}\label{eq:AH_and_H_commute}
A_H \circ \mathcal{H} = \mathcal{H} \circ A_H \;\;\; ,
\end{equation}
where we recall that $\mathcal{H}$ projects onto the horizontal subspace.

Using $\pi_*$ we can push-forward any $X\in TG$ to $TM$, and using $\pi^*$ we can pull-back any $T\in \mathfrak{T}^0_p M$. We can also lift any vector field $X\in TM$ from $TM$ to $TG$ with $\mathfrak{L}(\cdot )$. We now define the analogous map that takes some $T\in \mathfrak{T}^0_p G$ to a $(0,p)$-tensor field in $\mathfrak{T}^0_p M$:
\begin{definition}
For any $f\in \Omega^0(G)$, the $\pi_H$-projected function $\pi_H f \in \Omega^0(M)$ is the unique function on $M$ that satisfies
\begin{equation}\label{eq:pi_H_on_functions}
\pi^* \pi_H f = A_H f \;\;\; .
\end{equation}
\end{definition}
\noindent The value of $\pi_H f$ at any point $g\eta\in M$ is then the average of $f$ along the fiber above $g\eta$. Similarly, we have:
\begin{definition}\label{def:pi_H_proj_for_tensor}
For any $T \in \mathfrak{T}^0_p G$, the $\pi_H$-projected tensor field $\pi_H T \in \mathfrak{T}^0_p M$ is the unique tensor field on $M$ that satisfies
\begin{equation}
\pi^* \pi_H T =  \mathcal{H} A_H T \;\;\; .
\end{equation}
\end{definition}
\noindent It is straightforward to verify the uniqueness of $\pi_H T$ using the isomorphism between $TM$ and horizontal $R_H$-invariant vector fields in $TG$.

\subsubsection{Covariant Differentiation}

Consider the Levi-Civita connection, $\nabla$, on $G$. One can see that $\nabla_{\mathfrak{L}(X)}\mathfrak{L}(Y)$ is $R_H$-invariant for any $X, Y \in TM$ as follows:
\begin{align}
{R_h}_* \, \nabla_{\mathfrak{L}(X)}\mathfrak{L}(Y) & = \nabla_{{R_h}_*\, \mathfrak{L}(X)} {R_h}_* \, \mathfrak{L}(Y) \nonumber
\\
& = \nabla_{\mathfrak{L}(X)} \mathfrak{L}(Y) \;\;\; ,
\end{align}
where the first line follows from the fact that $R_h$ is an isometry on $G$, and the second line follows from the $R_H$-invariance of $\mathfrak{L}(X)$ and $\mathfrak{L}(Y)$. This allows us to make the following
\begin{definition}\label{def:cov_deriv_homog_def}
For any two vector fields $X, Y \in TM$, we define the connection on $M$, also denoted by $\nabla$, at a point $g\eta \in M$, as
\begin{equation}\label{eq:cov_deriv_homog_def}
\nabla_X Y \big|_{g\eta} := \pi_*  \left( \nabla_{\mathfrak{L}(X)} \mathfrak{L}(Y) \big|_g \right) \;\;\; ,
\end{equation}
where the RHS is independent of the choice of $g$ in the fiber from the previously established $R_H$-invariance of $\nabla_{\mathfrak{L}(X)} \mathfrak{L}(Y)$.
\end{definition}
\noindent In Appendix~\ref{sec:levi-civita_connection} we show that $\nabla$, on $M$, is the Levi-Civita connection on $M$ with respect to the $G$-invariant metric induced on $M$ via the bi-invariant metric on $G$. The covariant derivative of any tensor field on $M$ can then be defined in the usual way. We also define the symmetrised covariant derivative on functions and 1-forms in the analogous way to the definition,~\eqref{eq:def_symm_cov_deriv}, on $G$.

From the above definition, it is straightforward to see that
\begin{equation}\label{eq:cov_deriv_homog_lift}
\mathfrak{L}( \nabla_X Y ) = \mathcal{H} \nabla_{\mathfrak{L}(X)} \mathfrak{L}(Y) \;\;\; .
\end{equation}
One can also show that $\pi^* \nabla f = \nabla \pi^* f$ for any $f \in \Omega^0(M)$ as follows. For any $X\in TG$ we have
\begin{align}\label{eq:cov_pi_pull_back_commute_on_f}
(\pi^* \nabla f)(X)\big|_g & = (\nabla f)(\pi_* X)\big|_{g\eta} \nonumber
\\
& = ( \pi_* X )_{g \eta}(f) \nonumber
\\
& = X_g ( \pi^* f ) \nonumber
\\
& = (\nabla \pi^* f )(X)\big|_g \;\;\; .
\end{align}
As $\nabla f= d f$, this is simply the statement that the pull-back commutes with $d$.

For any $\omega \in \Omega^1(G)$, and any horizontal vector fields $X,Y \in \mathcal{H}TG$, we have the identity
\begin{equation}\label{eq:sym_cov_deriv_commutes_with_horizontal_proj}
(\nabla \omega )(X,Y) = (\nabla \mathcal{H} \omega )(X,Y) \;\;\; ,
\end{equation}
which can be seen as follows. Recall that $X = X_L^i \, l_i$ for any $X \in \mathcal{H}TG$, and similarly $\mathcal{H}\omega = (\omega_L)_i \, l^i$ for any $\omega\in\Omega^1(G)$. We have
\begin{align}
(\nabla \omega )(X,Y) & = X((\omega_L)_i)Y_L^i + Y((\omega_L)_i)X_L^i \nonumber
\\
& = ( \nabla \mathcal{H} \omega )( X, Y)
\end{align}
where we have used~\eqref{eq:sym_cov_deriv_omega_components} in line 1. 

For any $\omega\in\Omega^1(M)$, the symmetrised covariant derivative on $M$ also takes on a simple form. To see this we first note that, for any $X,Y\in TM$,
\begin{align}
(\nabla_X \omega )(Y)\big|_{g\eta} & = X(\omega(Y))\big|_{g\eta} - \omega( \nabla_X Y )\big|_{g\eta}
\nonumber
\\
& = \mathfrak{L}(X)\left( (\pi^* \omega )(\mathfrak{L}(Y)) \right) \big|_{g} - (\pi^* \omega )\left( \nabla_{\mathfrak{L}( X)} \mathfrak{L}( Y) \right)\big|_g
\nonumber
\\
& = ( \nabla_{\mathfrak{L}(X)}\pi^* \omega )(\mathfrak{L}(Y))\big|_g \;\;\; ,
\end{align}
which means that, after symmetrising over any pair $X,Y \in TM$, we have
\begin{equation}
(\nabla \omega )(X,Y)\big|_{g\eta} = \left( \nabla ( \pi^* \omega )\right)( \mathfrak{L}(X) , \mathfrak{L}(Y) )\big|_g \;\;\; .
\end{equation}
We also have the following derivation:
\begin{align}
(\nabla \omega )(X,Y)\big|_{g\eta} & = (\nabla \omega )(\pi_* \mathfrak{L}(X) , \pi_* \mathfrak{L}(Y))\big|_{g\eta}
\nonumber
\\
& = \left( \pi^*\nabla \omega \right)( \mathfrak{L}(X) ,  \mathfrak{L}(Y))\big|_g \;\;\; ,
\end{align}
which implies that
\begin{equation}\label{eq:pi_nabla_commute_on_omega_on_lifted_vecs}
\pi^* \nabla \omega = \nabla \pi^* \omega \;\;\; ,
\end{equation}
when acting on lifted vector fields $\mathfrak{L}(TM)\subset TG$.

\subsection{Convolution}\label{convhom}

\subsubsection{Definition}

Consider a compact manifold $M$ on which a compact Lie group $G$ acts transitively. We make the following
\begin{definition}\label{def:conv_def_homogeneous}
For any $S \in \mathfrak{T}^0_p M$ and any $T\in \mathfrak{T}^0_q M$, we define the convolution $ S \otimes_{AB} T \in \mathfrak{T}^0_{p+q}M$, for $AB=RL,LR$, as
\begin{equation}\label{eq:conv_def_homogeneous}
S \otimes_{AB} T := \pi_H \left( \pi^* S \otimes_{AB} \pi^* T \right) \;\;\; .
\end{equation}
where the convolution $\pi^* S \otimes_{AB} \pi^* T$ on the RHS is the convolution defined on $G$ in \textbf{Def}.~\ref{def:tensor_convolution_definition_G}.
\end{definition}
\noindent In words, the convolution of two fields on $M$ consists of first pulling them back to the group manifold $G$, convolving them on $G$, averaging over the fiber $H$, then projecting the result back down to $M$.

Recall that the $\pi_H$-projection involves averaging over the fiber and projecting onto the horizontal component. In fact, the averaging is not necessary as $ \pi^* S \otimes_{AB} \pi^* T$ is $R_H$-invariant. To see this, we focus on the $\otimes_{RL}$ and $\otimes_{LR}$ convolutions separately.

For any $h\in H$ we have
\begin{align}
R_h^* ( \pi^* S \otimes_{RL} \pi^* T ) & = \int_G dg\, R_h^* \mathfrak{r}_g \pi^* S \otimes R_h^* L_{g^{-1}}^* \pi^* T
\nonumber 
\\
& = \int_G dg\, \mathfrak{r}_g \pi^* S \otimes L_{g^{-1}}^* R_h^*  \pi^* T
\nonumber
\\
& = \int_G dg\, \mathfrak{r}_g \pi^* S \otimes L_{g^{-1}}^* \pi^* T
\nonumber
\\
& = \pi^* S \otimes_{RL} \pi^* T \;\;\; ,
\end{align}
where from line 1 to 2 we used the fact that $R_h^* \mathfrak{r}_g = \mathfrak{r}_g$ (as $\mathfrak{r}_g$ already creates a right-invariant field), and the fact that right and left pull-backs commute. From line 2 to 3 we used the fact that any pulled-back field, e.g. $\pi^* T$, is $R_H$-invariant.

For the $\otimes_{LR}$ convolution we have
\begin{align}
R_h^* ( \pi^* S \otimes_{LR} \pi^* T ) & = \int_G dg\, R_h^* \mathfrak{l}_g \pi^* S \otimes R_h^* R_{g^{-1}}^* \pi^* T
\nonumber 
\\
& = \int_G dg\, R_h^* \mathfrak{l}_g \pi^* S \otimes  R_{hg^{-1}}^* \pi^* T \;\;\; ,
\end{align}
using the fact that $R_g^* R_{g'}^* = R_{gg'}^*$. In order to proceed we need to determine how $R_h^*$ and $\mathfrak{l}_g$ commute. Consider some $\omega\in\Omega^1(G)$, then $R_h^* \mathfrak{l}_g \omega$ evaluated at $g'\in G$ is
\begin{align}
( R_h^* \mathfrak{l}_g \omega )\big|_{g'} & = ( R_h^* ( \mathfrak{l}_g \omega )\big|_{g'h} )\big|_{g'}
\nonumber
\\
& = ( R_h^* ( L^*_{gh^{-1}{g'}^{-1}} \omega )\big|_{g'h} )\big|_{g'} 
\nonumber
\\
& = ( R_h^* L^*_{gh^{-1}{g'}^{-1}} \omega )\big|_{g'}
\nonumber
\\
& = ( L^*_{gh^{-1}{g'}^{-1}} ( R_h^* \omega ) )\big|_{g'}
\nonumber
\\
& = ( \mathfrak{l}_{gh^{-1}} R_h^* \omega )\big|_{g'} \;\;\; ,
\end{align}
where from line 1 to 2 we used the definition of $\mathfrak{l}_g$ in~\eqref{eq:definition_mathfrak_lr_maps}. In summary, $R_h^* \mathfrak{l}_g = \mathfrak{l}_{gh^{-1}}R_h^*$. This then means that
\begin{align}
R_h^* ( \pi^* S \otimes_{LR} \pi^* T ) & = \int_G dg\, \mathfrak{l}_{gh^{-1}} R_h^* \pi^* S \otimes  R_{hg^{-1}}^* \pi^* T
\nonumber 
\\
& = \int_G d\tilde{g}\, \mathfrak{l}_{\tilde{g}} R_h^* \pi^* S \otimes  R_{\tilde{g}^{-1}}^* \pi^* T 
\nonumber
\\
& = \pi^* S \otimes_{LR} \pi^* T \;\;\; ,
\end{align}
where from line 1 to 2 we changed integration variables to $\tilde{g} = gh^{-1}$.

Now that we have established the $R_H$ invariance of $\pi^* S \otimes_{AB} \pi^* T$, we get the following simplification of the convolution $S \otimes_{AB} T$ on $M$, when acting on test vector fields $X_1 , X_2, ... \in TM$:
\begin{align}\label{eq:simplification_tensor_convolution_on_M}
( S \otimes_{AB} T )(X_1 , ... ) & = ( S \otimes_{AB} T )(\pi_* \mathfrak{L}( X_1 ) , ... )
\nonumber
\\
& = (\pi^*( S \otimes_{AB} T ))(\mathfrak{L}( X_1 ) , ... )
\nonumber
\\
& = (\pi^* \pi_H ( \pi^* S \otimes_{AB} \pi^* T ))(\mathfrak{L}( X_1 ) , ... )
\nonumber
\\
& = ( \mathcal{H}A_H ( \pi^* S \otimes_{AB} \pi^* T ) ) (\mathfrak{L}( X_1 ) , ... )
\nonumber
\\
& = ( \pi^* S \otimes_{AB} \pi^* T )(\mathfrak{L}( X_1 ) , ... ) \;\;\; ,
\end{align}
where we have used~\eqref{eq:lift_identities_1} in line 1, \textbf{Def}.~\eqref{def:conv_def_homogeneous} of the convolution on $M$ from line 2 to 3, and \textbf{Def}.~\eqref{def:pi_H_proj_for_tensor} of $\pi_H$ from line 3 to 4. To get the final line we used the fact that the horizontal projection $\mathcal{H}$ can be removed as the tensor is already acting on horizontal lifted vector fields, and the fact that $A_H$ acts trivially as $\pi^* S \otimes_{AB} \pi^* T$ is $R_H$-invariant.

In words, we have just shown that the convolution of two tensors, $S$ and $T$ on $M$, acts on vector fields on $M$ as the convolution of the pulled-back tensors, $\pi^*S$ and $\pi^*T$ on $G$, on the corresponding lifted vector fields. For the convolution of a pair of functions, $f\vee_{AB}\tilde{f}$, this simplifies to the statement that the value of $f\vee_{AB}\tilde{f}$ at the point $g\eta\in M$ is equal to the value of $\pi^* f \vee_{AB} \pi^* \tilde{f}$ at any point in the fiber above $g\eta$.

\subsubsection{Differentiation}\label{diffhom}

In this section we will be concerned with the symmetrised covariant derivative of the symmetrised convolution $\vee_{AB}$. For our purposes, we only need to consider the convolution $S\vee_{AB}T$, where $S$ and $T$ are either both functions, or one of them is a function and the other is a 1-form. The symmetrised covariant derivative of the convolution, $\nabla ( S\vee_{AB}T )$, is then either a 1-form or a symmetric $(0,2)$-tensor respectively, and hence acts on either one or two vector fields. We write $( \nabla ( S\vee_{AB}T )) (X_1 , ... )$, where $X_1 , ... \in TM$, to cover both these cases.

We have
\begin{align}
( \nabla ( S\vee_{AB}T )) (X_1 , ... ) & = ( \nabla ( S\vee_{AB}T ))( \pi_* \mathfrak{L}( X_1 ) , ... )
\nonumber
\\
& = ( \pi^* \nabla ( S\vee_{AB}T ) )( \mathfrak{L}( X_1 ) , ... )
\nonumber
\\
& = ( \nabla \pi^*  ( S\vee_{AB}T ) )( \mathfrak{L}( X_1 ) , ... )
\nonumber
\\
& = ( \nabla \pi^* \pi_H  ( \pi^* S\vee_{AB} \pi^* T ) )( \mathfrak{L}( X_1 ) , ... )
\nonumber
\\
& = ( \nabla \mathcal{H}A_H  ( \pi^* S\vee_{AB} \pi^* T ) )( \mathfrak{L}( X_1 ) , ... )
\nonumber
\\
& = ( \nabla ( \pi^* S\vee_{AB} \pi^* T ) )( \mathfrak{L}( X_1 ) , ... )
\;\;\; ,
\end{align}
where from line 2 to 3 we have used the fact that, for any function or 1-form $U$ on $M$, $\pi^* \nabla U = \nabla \pi^* U$ when acting on lifted vector fields (see~\eqref{eq:pi_nabla_commute_on_omega_on_lifted_vecs}). From line 5 to 6 we have used~\eqref{eq:sym_cov_deriv_commutes_with_horizontal_proj} to remove the horizontal projection $\mathcal{H}$. We also removed $A_H$ as it acts trivially on the $R_H$-invariant $\pi^* S\vee_{AB} \pi^* T$.

We can now utilise the derivative rules for the convolution on $G$ (\textbf{Prop}.~\eqref{prop:cov_deriv_property_on_G} and \textbf{Prop}.~\eqref{prop:cov_deriv_property_2_on_G}) to write
\begin{equation}\label{eq:almost_deriv_rule_on_M}
\nabla ( \pi^* S\vee_{AB} \pi^* T ) = \nabla \pi^* S\vee_{AB} \pi^* T = \pi^* S\vee_{AB} \nabla \pi^* T  \;\;\; .
\end{equation}
If $S$ is a function, we can use the fact that $\pi^*$ and $\nabla$ commute on functions to write
\begin{equation}
\nabla \pi^* S\vee_{AB} \pi^* T =  \pi^* \nabla S\vee_{AB} \pi^* T \;\;\; ,
\end{equation}
and hence
\begin{align}
( \nabla ( S\vee_{AB}T )) (X_1 , ... ) & = (\pi^* \nabla S\vee_{AB} \pi^* T)(\mathfrak{L}(X_1 ) ) 
\nonumber
\\
& = ( \nabla S\vee_{AB}T ) (X_1 , ... ) \;\;\; ,
\end{align}
using~\eqref{eq:simplification_tensor_convolution_on_M} to get the last line. For the same reason, $\nabla$ can be moved onto $T$ if $T$ is a function. This establishes the analogue of \textbf{Prop}.~\eqref{prop:cov_deriv_property_on_G} for the convolution on $M$.

The analogue of \textbf{Prop}.~\eqref{prop:cov_deriv_property_2_on_G} is more complicated. Following \textbf{Prop}.~\eqref{prop:cov_deriv_property_2_on_G} we take $S\in\Omega^0(M)$ and $T\in\Omega^1(M)$. From~\eqref{eq:almost_deriv_rule_on_M} we then have
\begin{align}
( \nabla ( S\vee_{AB} T ))( X_1 , ...) & = (\pi^* S\vee_{AB} \nabla \pi^* T )(\mathfrak{L}(X_1) , ... ) 
\nonumber
\\
& = \int_G dg \, ( \pi^* S )(g) ( B_{g^{-1}}^* \nabla \pi^* T )(\mathfrak{L}(X_1) , ... )  \;\;\; .
\end{align}
If we take $B=L$, i.e. we are considering the $\vee_{RL}$ convolution, then
\begin{align}
( L_{g^{-1}}^* \nabla \pi^* T )(\mathfrak{L}(X_1) , ... ) & = (  \nabla \pi^* T )( {L_{g^{-1}}}_* \mathfrak{L}(X_1) , ... ) 
\nonumber
\\
& = ( \pi^* \nabla  T )( {L_{g^{-1}}}_* \mathfrak{L}(X_1) , ... )
\nonumber
\\
& = ( L_{g^{-1}}^* \pi^* \nabla  T )( \mathfrak{L}(X_1) , ... )  \;\;\; .
\end{align}
To go from line 1 to 2 we have used the fact that the left push-forward of a lifted vector field is still a lifted vector field (that is, it is horizontal and $R_H$-invariant), and hence we can apply~\eqref{eq:pi_nabla_commute_on_omega_on_lifted_vecs} to commute $\nabla$ and $\pi^*$. This then means that
\begin{align}
( \nabla ( S\vee_{RL} T ))( X_1 , ...) & =  \int_G dg \, ( \pi^* S )(g) ( L_{g^{-1}}^* \pi^* \nabla T )(\mathfrak{L}(X_1) , ... )
\nonumber
\\
& = ( \pi^* S \vee_{RL} \pi^* \nabla T )(\mathfrak{L}(X_1) , ... ) 
\nonumber
\\
& = ( S\vee_{RL} \nabla T )( X_1 , ...)
\;\;\; ,
\end{align}
which establishes the analogue of \textbf{Prop}.~\eqref{prop:cov_deriv_property_2_on_G} for the $\vee_{RL}$ convolution on $M$. From the symmetry of $\vee_{AB}$, we also have $\nabla ( T \vee_{LR} S ) = \nabla T \vee_{LR} S$, where $S\in\Omega^0(M)$ and $T\in\Omega^1(M)$.

If we take $B=R$, i.e. we are considering $\nabla( S\vee_{LR}T )$ with $S\in\Omega^0(M)$ and $T\in\Omega^1(M)$, we cannot make the same simplification, as the right push-forward of a lifted vector field is \emph{not} a lifted vector field. It is not clear then whether the derivative rule $\nabla ( S \vee_{LR} T ) = S \vee_{LR} \nabla T$ is true when $S\in\Omega^0(M)$ and $T\in\Omega^1(M)$. This potential asymmetry between the $LR$- and $RL$-convolutions may be expected for a homogeneous space, as the manifold $M$ is identified with the \emph{left} coset space $\lbrace gH \, | \, g\in G\rbrace$.

\section{Discussion}

We introduced convolution products for tensor fields on group manifolds and homogeneous spaces and demonstrated that for symmetric convolutions of scalars and 1-forms they satisfy the usual (non-Leibniz) derivative rule. This followed from the observation that group manifolds, and homogeneous spaces, come with a natural path-independent notion of transporting tensors in a   manner compatible with the covariant derivative. This allowed us to apply the convolutions  to the construction of double copy field dictionaries for the BRST complex of graviton theory, considered to linear order in perturbation theory on these background manifolds. The robustness under BRST transformations was shown to follow from the convolution properties. We extended the  construction to static universe backgrounds by the  addition of a time direction.

It was shown that the spin-2 gauge transformations of the free graviton are reproduced correctly via the convolution. However, one can also  generate the dilaton and 2-form Kalb-Ramond field,  in the latter case by taking the two Yang--Mills BRST complexes to be different. Moreover, $p$-forms are expected to appear in more general double-copy constructible theories, particularly in the presence of supersymmetry. We leave the construction of a convolution which can adequately describe these types of fields to future work. Encouragingly, this was shown to exist in the simple case of the 2-sphere in \cite{Borsten:2019prq}.   

Another natural question is whether one can formulate the convolution on backgrounds of cosmological interest, such as de Sitter (perhaps making use of the fact that it is foliated by spheres in global coordinates, for which we already know how to define double copy dictionaries), or anti-de Sitter, which would allow us to explore potential links to holography. An important extension would be to construct the gravitational theories to higher orders perturbatively on these backgrounds, as was done for flat backgrounds in \cite{Luna:2016hge,Borsten:2020xbt,Borsten:2020zgj, Borsten:2021hua}. Here one could exploit recent developments extending color-kinematics duality to these spaces \cite{Armstrong:2020woi,Albayrak:2020fyp}.

\acknowledgments

LB is supported by the Leverhulme Research Project Grant RPG-2018-329. SN is supported by STFC grant ST/T000686/1. IJ is supported by the DIAS Schr\"{o}dinger Scholarship.

\appendix

\section{Associativity}\label{sec:app_associativity}

We first note the identities
\begin{equation}
R_g^* \, R_h^* = R_{gh}^* \;\;\; , \;\;\; L_g^* \, L_h^* = L_{hg}^* \;\;\; ,
\end{equation}
for any $g, h \in G$. From the definition of $\mathfrak{r}_g$, we can also see that, for any $T \in \mathfrak{T}^0_p G$,
\begin{align}
( \mathfrak{r}_g \, \mathfrak{r}_h \, T )\big|_{g'} & = R_{{g'}^{-1}g}^* \, ( \mathfrak{r}_h \, T  )_g \nonumber
\\
& =  R_{{g'}^{-1}g}^* \,  R_{g^{-1}h}^* \,  T_h \nonumber
\\
& =  R_{{g'}^{-1}h}^* \, T_h \nonumber
\\
& = (\mathfrak{r}_h \, T )\big|_{g'} \;\;\; ,
\end{align}
for any $g , h \in G$. Similarly,
\begin{align}
( \mathfrak{l}_g \, \mathfrak{l}_h \, T )\big|_{g'} & = L_{g{g'}^{-1}}^* \, ( \mathfrak{l}_h \, T  )_g \nonumber
\\
& =  L_{g{g'}^{-1}}^* \,  L_{h g^{-1}}^* \,  T_h \nonumber
\\
& =  L_{h{g'}^{-1}}^* \, T_h \nonumber
\\
& = (\mathfrak{l}_h \, T )\big|_{g'} \;\;\; .
\end{align}
That is, have the identities
\begin{equation}
\mathfrak{r}_g \, \mathfrak{r}_h = \mathfrak{r}_h \;\;\; , \;\;\; \mathfrak{l}_g \, \mathfrak{l}_h = \mathfrak{l}_h \;\;\; .
\end{equation}
We will also need the following identity:
\begin{equation}
\mathfrak{r}_h \, L_{g^{-1}}^* = L_{g^{-1}}^* \, \mathfrak{r}_{g^{-1}h} \;\;\; .
\end{equation}
which can be seen by considering the action of both the left and right sides of the above equation on some $(0,p)$-tensor $T$. Take the LHS:
\begin{align}
( \mathfrak{r}_h \, L_{g^{-1}}^* \, T )\big|_{g'} & = R_{{g'}^{-1}h}^* ( L_{g^{-1}}^* \, T  )\big|_h \nonumber
\\
& = R_{{g'}^{-1}h}^* L_{g^{-1}}^* \, T_{g^{-1}h} \;\;\; ,
\end{align}
and then the RHS:
\begin{align}
( L_{g^{-1}}^*  \, \mathfrak{r}_{g^{-1}h} \, T )\big|_{g'} & = L_{g^{-1}}^* ( \mathfrak{r}_{g^{-1}h} \, T  )\big|_{g^{-1}g'} \nonumber
\\
& = L_{g^{-1}}^* \,  R_{{g'}^{-1}g g^{-1} h}^* \, T_{g^{-1}h} \nonumber
\\
& =  R_{{g'}^{-1}g g^{-1} h}^* \, L_{g^{-1}}^* \, T_{g^{-1}h} \nonumber
\\
& =  R_{{g'}^{-1} h}^* \, L_{g^{-1}}^* \, T_{g^{-1}h} \;\;\; .
\end{align}
We are now ready to prove the associativity of $\otimes_{RL}$.

\noindent \textit{Proof}
\begin{align}
S \otimes_{RL} ( T \otimes_{RL} U ) & = \int_G dg  \, \mathfrak{r}_g S \otimes L_{g^{-1}}^* ( T \otimes_{RL} U )  \nonumber
\\
& = \int_G dg  \, \mathfrak{r}_g S \otimes L_{g^{-1}}^* \left( \int_G dh \, \mathfrak{r}_h T \otimes L_{h^{-1}}^* U \right)  \nonumber
\\
& = \int_{G \times G} dg \, dh \, \mathfrak{r}_g S \otimes L_{g^{-1}}^*  \mathfrak{r}_h T \otimes L_{g^{-1}}^* L_{h^{-1}}^* U \nonumber
\\
& = \int_{G \times G} dg \, dh \, \mathfrak{r}_g S \otimes L_{g^{-1}}^*  \mathfrak{r}_h T \otimes L_{h^{-1}g^{-1}}^* U \;\;\; ,
\end{align}
where we have used the fact that the pull-back distributes onto both sides of a tensor product. If we change integration variables from $h$ to $\tilde{h} = gh$, we get
\begin{align}
S \otimes_{RL} ( T \otimes_{RL} U ) & = \int_{G \times G} dg \, d\tilde{h} \, \mathfrak{r}_g S \otimes L_{g^{-1}}^*  \mathfrak{r}_{g^{-1}\tilde{h}} T \otimes L_{\tilde{h}^{-1}}^* U \nonumber
\\
& =  \int_{G \times G} dg \, d\tilde{h} \, \mathfrak{r}_g S \otimes \mathfrak{r}_{\tilde{h}}  L_{g^{-1}}^* T \otimes L_{\tilde{h}^{-1}}^* U \nonumber
\\
& = \int_{G \times G} dg \, d\tilde{h} \, \mathfrak{r}_{\tilde{h}} \left(  \mathfrak{r}_g S \otimes L_{g^{-1}}^* T \right) \otimes L_{\tilde{h}^{-1}}^* U \nonumber
\\
& = \int_G d\tilde{h} \,  \mathfrak{r}_{\tilde{h}} \left( \int_G dg \, \mathfrak{r}_g S \otimes L_{g^{-1}}^* T \right) \otimes L_{\tilde{h}^{-1}}^* U \nonumber
\\
& = ( S \otimes_{RL} T ) \otimes_{RL} U \;\;\; .
\end{align} 
A similar calculation can be done to prove the associativity of $\otimes_{LR}$.

\section{Levi-Civita connection on $M$}\label{sec:levi-civita_connection}

For clarity, in this appendix we denote the bi-invariant metric on $G$ as $\tilde{\mathtt{g}}$, and the corresponding Levi-Civita connection on $G$ as $\tilde{\nabla}$. Recall that bi-invariance of $\tilde{\mathtt{g}}$ means that $L^*_g \tilde{\mathtt{g}} = R^*_{g'} \tilde{\mathtt{g}} = \tilde{\mathtt{g}}$, for any $g, g' \in G$.

$\tilde{\mathtt{g}}$ induces a metric, $\mathtt{g}$, on the base space $M$ in the following way. For any point $g\eta\in M$, and any $X_{g\eta}, Y_{g\eta} \in T_{g\eta}M$, we find the horizontal lifts $\mathfrak{L}(X)_g , \mathfrak{L}(Y)_g \in T_g G$, and define
\begin{equation}
\mathtt{g}(X_{g\eta}, Y_{g\eta}) := \tilde{\mathtt{g}}(\mathfrak{L}(X)_g , \mathfrak{L}(Y)_g) \;\;\; .
\end{equation}
Crucially, the RHS is independent of the point $g$ in the fiber above $g\eta$. This follows from the right-invariance of $\tilde{\mathtt{g}}$ and the $R_H$-invariance of the horizontal lifts.

The left action of $G$ on $M$, $\sigma_g \, : \, M \rightarrow M$, $x\rightarrow \sigma_g (x) = gx = g g' \eta$ (for any $g'\in G$ such that $g' \eta = x$), is a diffeomorphism of $M$. Given the corresponding pull-back, $\sigma_g^*$, we say that $\mathtt{g}$ is \emph{$G$-invariant} if $\sigma_g^* \mathtt{g} = \mathtt{g}$ for all $g\in G$. For completeness, we note the following
\begin{prop}
The induced metric, $\mathtt{g}$, is $G$-invariant.
\end{prop}
\begin{proof}
First, consider the pull-back $\sigma_g^*$ acting on a function $f\in \Omega^0(M)$. For some point $g'\eta \in M$ we have
\begin{align}
(\sigma_g^* f )(g'\eta) & = f(\sigma_g(g'\eta )) \nonumber
\\
& = f(gg'\eta ) \nonumber
\\
& = (\pi^* f)(g g') \nonumber
\\
& = (L_g^* \pi^* f)(g') \;\;\; .
\end{align}
We also have
\begin{equation}
(\sigma_g^* f )(g'\eta) = ( \pi^* \sigma_g^* f )(g' )\;\;\; ,
\end{equation}
and hence
\begin{equation}
\pi^* \sigma_g^* f = L_g^* \pi^* f \;\;\; ,
\end{equation}
for any $g\in G$ and any function $f\in\Omega^0(M)$. Following the definition of the push-forward $\pi_*$ on a vector field $X\in TG$, one can verify that 
\begin{equation}
{\sigma_g}_* \pi_* X = \pi_* {L_g}_* X \;\;\; .
\end{equation}
Similarly, for any $T\in \mathfrak{T}^0_p M$, one finds
\begin{equation}
\pi^* \sigma_g^* T = L_g^* \pi^* T \;\;\; .
\end{equation}

Now, consider some $g'\in G$, and two vectors $X_{g'\eta}, Y_{g'\eta} \in T_{g'\eta}M$. We have
\begin{align}\label{eq:g_inv_derivation_1}
(\sigma_g^* \mathtt{g})(X_{g'\eta} , Y_{g'\eta}) & = (\sigma_g^* \mathtt{g})(
{( \pi_* { \mathfrak{L}{( X_{g'\eta} )}  }_{g'}  )}_{g'\eta}
,
{( \pi_* {  \mathfrak{L}{( Y_{g'\eta}  )}_{g'}   }  )}_{g'\eta}
) \nonumber
\\
& = (\pi^* \sigma_g^* \mathtt{g})(
{ \mathfrak{L}{( X_{g'\eta} )}  }_{g'}
,
{  \mathfrak{L}{( Y_{g'\eta}  )}_{g'}   } 
) \nonumber
\\
& = (L_g^* \pi^* \mathtt{g} )(
{ \mathfrak{L}{( X_{g'\eta} )}  }_{g'}
,
{  \mathfrak{L}{( Y_{g'\eta}  )}_{g'}   } 
) \nonumber
\\
& = (\pi^* \mathtt{g} )(
{( {L_g}_*{  \mathfrak{L}{( X_{g'\eta} )}  }_{g'} )}_{gg'}
,
{( {L_g}_*{  \mathfrak{L}{( Y_{g'\eta}  )}_{g'}   } )}_{gg'}
) \;\;\; ,
\end{align}
where we have used~\eqref{eq:lift_identities_1} in line 1. We then note that ${L_g}_* \mathfrak{L}(X)$ is horizontal for any $X \in TM$, and that, for any horizontal $\tilde{X}_g , \tilde{Y}_g \in \mathcal{H}T_gG$, 
\begin{equation}\label{eq:pullback_of_g_and_tilde_metric_relation}
(\pi^* \mathtt{g})( \tilde{X}_g , \tilde{Y}_g ) = \tilde{\mathtt{g}}( \tilde{X}_g , \tilde{Y}_g ) \;\;\; ,
\end{equation}
which can be seen as follows:
\begin{align}
(\pi^* \mathtt{g})( \tilde{X}_g , \tilde{Y}_g ) & = \mathtt{g}( \pi_* \tilde{X}_g ,\pi_* \tilde{Y}_g ) \nonumber
\\
& = \tilde{\mathtt{g}}({\mathfrak{L}(\pi_* \tilde{X}_g)  )}_{g} , {\mathfrak{L}(\pi_* \tilde{Y}_g)  )}_{g} ) \nonumber
\\
& =  \tilde{\mathtt{g}}( \mathcal{H}\tilde{X}_g ,  \mathcal{H}\tilde{Y}_g ) \;\;\; , \nonumber
\\
& =  \tilde{\mathtt{g}}( \tilde{X}_g ,  \tilde{Y}_g ) \;\;\; ,
\end{align}
where line 2 follows from the definition of $\mathtt{g}$, and line 3 follows from~\eqref{eq:lift_identities_2}, and line 4 follows from the fact that $\tilde{X}_g$ and $\tilde{Y}_g$ are already horizontal. Following the last line in~\eqref{eq:g_inv_derivation_1}, this implies that
\begin{align}
(\sigma_g^* \mathtt{g})(X_{g'\eta} , Y_{g'\eta}) & = \tilde{\mathtt{g}}( {( {L_g}_*{  \mathfrak{L}{( X_{g'\eta} )}  }_{g'} )}_{gg'} ,  {( {L_g}_*{  \mathfrak{L}{( Y_{g'\eta} )}  }_{g'} )}_{gg'} ) \nonumber
\\
& = \tilde{\mathtt{g}}(
{  \mathfrak{L}{( X_{g'\eta} )}  }_{g'} , {  \mathfrak{L}{( Y_{g'\eta} )}  }_{g'}
) \nonumber
\\
& = \mathtt{g}(X_{g'\eta} , Y_{g'\eta}) \;\;\; ,
\end{align}
where line 2 follows from the left-invariance of $\tilde{\mathtt{g}}$, and line 3 follows from the definition of $\mathtt{g}$.
\end{proof}

$\tilde{\nabla}$ is torsionless and compatible with $\tilde{g}$. Respectively,
\begin{align}
\tilde{T}(X,Y) & =  \tilde{\nabla}_X Y - \tilde{\nabla}_Y X - [X,Y] = 0 \;\;\; , 
\\
(\tilde{\nabla}_{X}\tilde{\mathtt{g}})(Y,Z) & = X\big( \tilde{\mathtt{g}}(Y,Z) \big) - \tilde{\mathtt{g}}( \tilde{\nabla}_X Y , Z ) - \tilde{\mathtt{g}} ( Y , \tilde{\nabla}_X Z ) = 0 \;\;\; ,
\end{align}
for any $X,Y,Z \in TG$. Using the above, and the definitions of $\mathtt{g}$ and $\nabla$, we now show that
\begin{prop}
$\nabla$ is torsionless. That is, for any $X,Y \in TM$,
\begin{equation}
T(X,Y) =  \nabla_X Y - \nabla_Y X - [X,Y] = 0 \;\;\; .
\end{equation}
\end{prop}
\begin{proof}
Take any $g\in G$ and the corresponding base point $g\eta \in M$. Using the definition of $\nabla$ we have
\begin{align}
\nabla_X Y \big|_{g\eta} - \nabla_Y X \big|_{g\eta} & = \pi_* \Big(\tilde{\nabla}_{\mathfrak{L}(X)}\mathfrak{L}(Y) \big|_g - \tilde{\nabla}_{\mathfrak{L}(Y)} \mathfrak{L}(X) \big|_g \Big) \nonumber
\\
& = \pi_* \Big( [ \mathfrak{L}(X) , \mathfrak{L}(Y) ]\big|_g \Big) \;\;\; ,
\end{align}
where the last line follows from the torsionlessness of $\tilde{\nabla}$. Given that $\pi$ is a submersion, the vector 
\begin{equation}
[ \mathfrak{L}(X) , \mathfrak{L}(Y) ] - \mathfrak{L}(  [X,Y]) \; ,
\end{equation}
is vertical~\cite{gallot2004riemannian} (\textit{Lemma 3.54}), and hence it vanishes under $\pi_*$. This implies that 
\begin{align}
\pi_* \Big( [ \mathfrak{L}(X) , \mathfrak{L}(Y) ]\big|_g \Big) & = \pi_* \mathfrak{L}\big( [ \mathfrak{L}(X) , \mathfrak{L}(Y) ] \big)\big|_{g\eta} \nonumber
\\
& = [X,Y]\big|_{g\eta} \;\;\; ,
\end{align}
using~\eqref{eq:lift_identities_1} to get the last line. We have just shown that
\begin{equation}
\nabla_X Y \big|_{g\eta} - \nabla_Y X \big|_{g\eta} = [ X, Y]\big|_{g\eta} \;\;\; ,
\end{equation}
and hence $T(X,Y)=0$.
\end{proof}
Next, we show that
\begin{prop}
$\nabla$ is compatible with the metric $\mathtt{g}$. That is, for any $X,Y, Z \in TM$,
\begin{equation}\label{eq:nabla_compatible_with_g}
(\nabla_{X}\mathtt{g})(Y,Z) = X\big( \mathtt{g}(Y,Z) \big) - \mathtt{g}( \nabla_X Y , Z ) - \mathtt{g} ( Y , \nabla_X Z ) = 0 \;\;\; .
\end{equation}
\end{prop}
\begin{proof}
We first note that, for any $g\in G$, any $X_1 , ... , X_p \in TM$, and any $T\in \mathfrak{T}^0_p M$,
\begin{equation}
\pi^* \big( T(X_1 , ..., X_p ) \big)\big|_{g\eta} = (\pi^* T)(\mathfrak{L}(X_1 ), ... ,\mathfrak{L}(X_p ) )\big|_g \;\;\; ,
\end{equation}
where the LHS is the pull-back of the function $T(X_1 , ..., X_p )$ on $M$. For any $X,Y \in TM$, we can then use~\eqref{eq:pullback_of_g_and_tilde_metric_relation} to write the pull-back of the function $\mathtt{g}(X,Y)$ on $M$ as
\begin{equation}
\pi^* \big( g(X , Y ) \big)\big|_{g\eta} = \tilde{g}(\mathfrak{L}(X ),\mathfrak{L}(Y ))\big|_g \;\;\; .
\end{equation}
In~\eqref{eq:nabla_compatible_with_g} we can then write
\begin{align}
X(\mathtt{g}(Y,Z))\big|_{g\eta} &  = \pi_* \mathfrak{L}(X)(\mathtt{g}(Y,Z))\big|_{g\eta} \nonumber
\\
& = \mathfrak{L}(X)\Big( \pi^* \big( \mathtt{g}(Y,Z) \big) \Big) \Big|_g \nonumber
\\
& = \mathfrak{L}(X)\big( \tilde{g}(\mathfrak{L}(X ),\mathfrak{L}(Y )) \big) \big|_g \;\;\; .
\end{align}
Focussing on the next term in~\eqref{eq:nabla_compatible_with_g} we have
\begin{align}
\mathtt{g}( \nabla_X Y , Z )\big|_{g\eta} & = \mathtt{g}( \pi_* \tilde{\nabla}_{\mathfrak{L}(X)} \mathfrak{L}(Y) , \pi_* \mathfrak{L}(Z) )\big|_{g\eta} \nonumber
\\
& = ( \pi^* \mathtt{g}) \big( \mathcal{H}\tilde{\nabla}_{\mathfrak{L}(X)} \mathfrak{L}(Y) , \mathfrak{L}(Z) \big)\big|_g \nonumber
\\
& = \tilde{\mathtt{g}}\big( \mathcal{H}\tilde{\nabla}_{\mathfrak{L}(X)} \mathfrak{L}(Y) , \mathfrak{L}(Z) \big)\big|_g \nonumber
\\
& = \tilde{\mathtt{g}}\big( \tilde{\nabla}_{\mathfrak{L}(X)} \mathfrak{L}(Y) , \mathfrak{L}(Z) \big)\big|_g \;\;\; ,
\end{align}
where to get line 1 we use the definition of $\nabla$ and~\eqref{eq:lift_identities_1}, to get line 3 we use~\eqref{eq:pullback_of_g_and_tilde_metric_relation}, and to get the last line we have used the fact that $\mathfrak{L}(Z)$ is horizontal, and therefore is orthogonal (with respect to the metric $\tilde{\mathtt{g}}$) to any vertical component of $\tilde{\nabla}_{\mathfrak{L}(X)} \mathfrak{L}(Y)$. For the last term in~\eqref{eq:nabla_compatible_with_g} we similarly find
\begin{equation}
\mathtt{g}(  Y , \nabla_X Z )\big|_{g\eta} = \tilde{\mathtt{g}}\big(  \mathfrak{L}(Y) ,\tilde{\nabla}_{\mathfrak{L}(X)} \mathfrak{L}(Z) \big)\big|_g  \;\;\; .
\end{equation} 
Putting it all together we find
\begin{align}
(\nabla_{X}\mathtt{g})(Y,Z)\big|_{g\eta} & = X\big( \mathtt{g}(Y,Z) \big)\big|_{g\eta} - \mathtt{g}( \nabla_X Y , Z )\big|_{g\eta} - \mathtt{g} ( Y , \nabla_X Z )\big|_{g\eta} \nonumber
\\
& = \mathfrak{L}(X)\big( \tilde{g}(\mathfrak{L}(X ),\mathfrak{L}(Y )) \big) \big|_g \nonumber
\\
& \;\;\;\;\;\; -\tilde{\mathtt{g}}\big( \tilde{\nabla}_{\mathfrak{L}(X)} \mathfrak{L}(Y) , \mathfrak{L}(Z) \big)\big|_g -\tilde{\mathtt{g}}\big(  \mathfrak{L}(Y) ,\tilde{\nabla}_{\mathfrak{L}(X)} \mathfrak{L}(Z) \big)\big|_g \nonumber
\\
& = (\tilde{\nabla}_{\mathfrak{L}(X)}\tilde{\mathtt{g}})(\mathfrak{L}(Y),\mathfrak{L}(Z))\big|_g \nonumber
\\
& = 0 \;\;\; ,
\end{align}
by the compatibility of $\tilde{\nabla}$ with $\tilde{\mathtt{g}}$.

\end{proof}

\begin{thebibliography}{10}

\bibitem{Borsten:2019prq}
L.~Borsten, I.~Jubb, V.~Makwana, and S.~Nagy, ``{Gauge × gauge on spheres}'',
  \href{http://dx.doi.org/10.1007/JHEP06(2020)096}{{\em JHEP} {\bf 06} (2020)
  096}, \href{http://arxiv.org/abs/1911.12324}{{\tt arXiv:1911.12324
  [hep-th]}}.

\bibitem{Anastasiou:2014qba}
A.~Anastasiou, L.~Borsten, M.~J. Duff, L.~J. Hughes, and S.~Nagy, ``{Yang-Mills
  origin of gravitational symmetries}'',
  \href{http://dx.doi.org/10.1103/PhysRevLett.113.231606}{{\em Phys.Rev.Lett.}
  {\bf 113} (2014) no.~23, 231606},
\href{http://arxiv.org/abs/1408.4434}{{\tt arXiv:1408.4434 [hep-th]}}.

\bibitem{Borsten:2017jpt}
L.~Borsten, ``{$D=6$, $\mathcal{N}=(2,0)$ and $\mathcal{N}=(4,0)$ theories}'',
  \href{http://dx.doi.org/10.1103/PhysRevD.97.066014}{{\em Phys. Rev.} {\bf
  D97} (2018)  066014},
\href{http://arxiv.org/abs/1708.02573}{{\tt arXiv:1708.02573 [hep-th]}}.

\bibitem{Anastasiou:2018rdx}
A.~Anastasiou, L.~Borsten, M.~J. Duff, S.~Nagy, and M.~Zoccali, ``{Gravity as
  Gauge Theory Squared: A Ghost Story}'',
  \href{http://dx.doi.org/10.1103/PhysRevLett.121.211601}{{\em Phys. Rev.
  Lett.} {\bf 121} (2018) no.~21, 211601},
\href{http://arxiv.org/abs/1807.02486}{{\tt arXiv:1807.02486 [hep-th]}}.

\bibitem{Zoccali:2018pty}
M.~Zoccali, {\em {Supergravity as Yang-Mills Squared}}.
\newblock PhD thesis, Imperial Coll., London,
2018.
\newblock

\bibitem{Borsten:2020xbt}
L.~Borsten and S.~Nagy, ``{The pure BRST Einstein-Hilbert Lagrangian from the
  double-copy to cubic order}'',
  \href{http://dx.doi.org/10.1007/JHEP07(2020)093}{{\em JHEP} {\bf 07} (2020)
  093}, \href{http://arxiv.org/abs/2004.14945}{{\tt arXiv:2004.14945
  [hep-th]}}.

\bibitem{Luna:2020adi}
A.~Luna, S.~Nagy, and C.~White, ``{The convolutional double copy: a case study
  with a point}'', \href{http://dx.doi.org/10.1007/JHEP09(2020)062}{{\em JHEP}
  {\bf 09} (2020)  062}, \href{http://arxiv.org/abs/2004.11254}{{\tt
  arXiv:2004.11254 [hep-th]}}.

\bibitem{Cardoso:2016amd}
G.~Cardoso, S.~Nagy, and S.~Nampuri, ``{Multi-centered $ \mathcal{N}=2 $ BPS
  black holes: a double copy description}'',
  \href{http://dx.doi.org/10.1007/JHEP04(2017)037}{{\em JHEP} {\bf 04} (2017)
  037},
\href{http://arxiv.org/abs/1611.04409}{{\tt arXiv:1611.04409 [hep-th]}}.

\bibitem{Cardoso:2016ngt}
G.~L. Cardoso, S.~Nagy, and S.~Nampuri, ``{A double copy for $ \mathcal{N}=2 $
  supergravity: a linearised tale told on-shell}'',
  \href{http://dx.doi.org/10.1007/JHEP10(2016)127}{{\em JHEP} {\bf 10} (2016)
  127},
\href{http://arxiv.org/abs/1609.05022}{{\tt arXiv:1609.05022 [hep-th]}}.

\bibitem{Dokmanic:2010}
I.~{Dokmanic} and D.~{Petrinovic}, ``Convolution on the $n$-sphere with
  application to pdf modeling'',
  \href{http://dx.doi.org/10.1109/TSP.2009.2033329}{{\em IEEE Transactions on
  Signal Processing} {\bf 58} (2010) no.~3, 1157--1170}.

\bibitem{Chakraborty:2018h}
R.~Chakraborty, M.~Banerjee, and B.~C. Vemuri, ``H-cnns: Convolutional neural
  networks for riemannian homogeneous spaces'', {\em arXiv preprint
  arXiv:1805.05487} (2018)  .

\bibitem{Cohen:2018general}
T.~Cohen, M.~Geiger, and M.~Weiler, ``A general theory of equivariant cnns on
  homogeneous spaces'', {\em arXiv preprint arXiv:1811.02017} (2018)  .

\bibitem{Bern:2019prr}
Z.~Bern, J.~J. Carrasco, M.~Chiodaroli, H.~Johansson, and R.~Roiban, ``{The
  Duality Between Color and Kinematics and its Applications}'',
\href{http://arxiv.org/abs/1909.01358}{{\tt arXiv:1909.01358 [hep-th]}}.

\bibitem{Borsten:2020bgv}
L.~Borsten, ``{Gravity as the square of gauge theory: a review}'',
  \href{http://dx.doi.org/10.1007/s40766-020-00003-6}{{\em Riv. Nuovo Cim.}
  {\bf 43} (2020) no.~3, 97--186}.

\bibitem{Kawai:1985xq}
H.~Kawai, D.~Lewellen, and S.~Tye, ``{A Relation Between Tree Amplitudes of
  Closed and Open Strings}'',
\href{http://dx.doi.org/10.1016/0550-3213(86)90362-7}{{\em Nucl.Phys.} {\bf
  B269} (1986)  1}.

\bibitem{Bern:1993wt}
Z.~Bern, D.~C. Dunbar, and T.~Shimada, ``{String based methods in perturbative
  gravity}'', \href{http://dx.doi.org/10.1016/0370-2693(93)91081-W}{{\em Phys.
  Lett.} {\bf B312} (1993)  277--284},
\href{http://arxiv.org/abs/hep-th/9307001}{{\tt arXiv:hep-th/9307001
  [hep-th]}}.

\bibitem{Bern:1998ug}
Z.~Bern, L.~J. Dixon, D.~C. Dunbar, M.~Perelstein, and J.~S. Rozowsky, ``{On
  the relationship between Yang-Mills theory and gravity and its implication
  for ultraviolet divergences}'',
  \href{http://dx.doi.org/10.1016/S0550-3213(98)00420-9}{{\em Nucl. Phys.} {\bf
  B530} (1998)  401--456},
\href{http://arxiv.org/abs/hep-th/9802162}{{\tt arXiv:hep-th/9802162
  [hep-th]}}.

\bibitem{Bern:2008qj}
Z.~Bern, J.~Carrasco, and H.~Johansson, ``{New Relations for Gauge-Theory
  Amplitudes}'', \href{http://dx.doi.org/10.1103/PhysRevD.78.085011}{{\em
  Phys.Rev.} {\bf D78} (2008)  085011},
\href{http://arxiv.org/abs/0805.3993}{{\tt arXiv:0805.3993 [hep-ph]}}.

\bibitem{Stieberger:2009hq}
S.~Stieberger, ``{Open \& Closed vs. Pure Open String Disk Amplitudes}'',
\href{http://arxiv.org/abs/0907.2211}{{\tt arXiv:0907.2211 [hep-th]}}.

\bibitem{BjerrumBohr:2009rd}
N.~E.~J. Bjerrum-Bohr, P.~H. Damgaard, and P.~Vanhove, ``{Minimal Basis for
  Gauge Theory Amplitudes}'',
  \href{http://dx.doi.org/10.1103/PhysRevLett.103.161602}{{\em Phys. Rev.
  Lett.} {\bf 103} (2009)  161602},
\href{http://arxiv.org/abs/0907.1425}{{\tt arXiv:0907.1425 [hep-th]}}.

\bibitem{BjerrumBohr:2010hn}
N.~E.~J. Bjerrum-Bohr, P.~H. Damgaard, T.~Sondergaard, and P.~Vanhove, ``{The
  Momentum Kernel of Gauge and Gravity Theories}'',
  \href{http://dx.doi.org/10.1007/JHEP01(2011)001}{{\em JHEP} {\bf 01} (2011)
  001},
\href{http://arxiv.org/abs/1010.3933}{{\tt arXiv:1010.3933 [hep-th]}}.

\bibitem{Feng:2010my}
B.~Feng, R.~Huang, and Y.~Jia, ``{Gauge Amplitude Identities by On-shell
  Recursion Relation in S-matrix Program}'',
  \href{http://dx.doi.org/10.1016/j.physletb.2010.11.011}{{\em Phys. Lett.}
  {\bf B695} (2011)  350--353},
\href{http://arxiv.org/abs/1004.3417}{{\tt arXiv:1004.3417 [hep-th]}}.

\bibitem{Chen:2011jxa}
Y.-X. Chen, Y.-J. Du, and B.~Feng, ``{A Proof of the Explicit Minimal-basis
  Expansion of Tree Amplitudes in Gauge Field Theory}'',
  \href{http://dx.doi.org/10.1007/JHEP02(2011)112}{{\em JHEP} {\bf 02} (2011)
  112}, \href{http://arxiv.org/abs/1101.0009}{{\tt arXiv:1101.0009 [hep-th]}}.

\bibitem{Mafra:2011kj}
C.~R. Mafra, O.~Schlotterer, and S.~Stieberger, ``{Explicit BCJ Numerators from
  Pure Spinors}'', \href{http://dx.doi.org/10.1007/JHEP07(2011)092}{{\em JHEP}
  {\bf 07} (2011)  092},
\href{http://arxiv.org/abs/1104.5224}{{\tt arXiv:1104.5224 [hep-th]}}.

\bibitem{Du:2016tbc}
Y.-J. Du and C.-H. Fu, ``{Explicit BCJ numerators of nonlinear simga model}'',
  \href{http://dx.doi.org/10.1007/JHEP09(2016)174}{{\em JHEP} {\bf 09} (2016)
  174},
\href{http://arxiv.org/abs/1606.05846}{{\tt arXiv:1606.05846 [hep-th]}}.

\bibitem{Mizera:2019blq}
S.~Mizera, ``{Kinematic Jacobi Identity is a Residue Theorem: Geometry of
  Color-Kinematics Duality for Gauge and Gravity Amplitudes}'',
  \href{http://dx.doi.org/10.1103/PhysRevLett.124.141601}{{\em Phys. Rev.
  Lett.} {\bf 124} (2020) no.~14, 141601},
  \href{http://arxiv.org/abs/1912.03397}{{\tt arXiv:1912.03397 [hep-th]}}.

\bibitem{Reiterer:2019dys}
M.~Reiterer, ``{A homotopy BV algebra for Yang-Mills and color-kinematics}'',
  \href{http://arxiv.org/abs/1912.03110}{{\tt arXiv:1912.03110 [math-ph]}}.

\bibitem{Bern:2017yxu}
Z.~Bern, J.~J. Carrasco, W.-M. Chen, H.~Johansson, and R.~Roiban, ``{Gravity
  Amplitudes as Generalized Double Copies of Gauge-Theory Amplitudes}'',
  \href{http://dx.doi.org/10.1103/PhysRevLett.118.181602}{{\em Phys. Rev.
  Lett.} {\bf 118} (2017) no.~18, 181602},
\href{http://arxiv.org/abs/1701.02519}{{\tt arXiv:1701.02519 [hep-th]}}.

\bibitem{Bern:2017ucb}
Z.~Bern, J.~J.~M. Carrasco, W.-M. Chen, H.~Johansson, R.~Roiban, and M.~Zeng,
  ``{Five-loop four-point integrand of $N=8$ supergravity as a generalized
  double copy}'', \href{http://dx.doi.org/10.1103/PhysRevD.96.126012}{{\em
  Phys. Rev.} {\bf D96} (2017) no.~12, 126012},
\href{http://arxiv.org/abs/1708.06807}{{\tt arXiv:1708.06807 [hep-th]}}.

\bibitem{Bern:2010ue}
Z.~Bern, J.~J.~M. Carrasco, and H.~Johansson, ``{Perturbative Quantum Gravity
  as a Double Copy of Gauge Theory}'',
  \href{http://dx.doi.org/10.1103/PhysRevLett.105.061602}{{\em Phys.Rev.Lett.}
  {\bf 105} (2010)  061602},
\href{http://arxiv.org/abs/1004.0476}{{\tt arXiv:1004.0476 [hep-th]}}.

\bibitem{Bern:2010tq}
Z.~Bern, J.~J.~M. Carrasco, L.~J. Dixon, H.~Johansson, and R.~Roiban, ``{The
  Complete Four-Loop Four-Point Amplitude in N=4 Super-Yang-Mills Theory}'',
  \href{http://dx.doi.org/10.1103/PhysRevD.82.125040}{{\em Phys. Rev.} {\bf
  D82} (2010)  125040},
\href{http://arxiv.org/abs/1008.3327}{{\tt arXiv:1008.3327 [hep-th]}}.

\bibitem{Carrasco:2011mn}
J.~J. Carrasco and H.~Johansson, ``{Five-Point Amplitudes in N=4
  Super-Yang-Mills Theory and N=8 Supergravity}'',
  \href{http://dx.doi.org/10.1103/PhysRevD.85.025006}{{\em Phys. Rev.} {\bf
  D85} (2012)  025006},
\href{http://arxiv.org/abs/1106.4711}{{\tt arXiv:1106.4711 [hep-th]}}.

\bibitem{Bern:2011rj}
Z.~Bern, C.~Boucher-Veronneau, and H.~Johansson, ``{N >= 4 Supergravity
  Amplitudes from Gauge Theory at One Loop}'',
  \href{http://dx.doi.org/10.1103/PhysRevD.84.105035}{{\em Phys. Rev.} {\bf
  D84} (2011)  105035},
\href{http://arxiv.org/abs/1107.1935}{{\tt arXiv:1107.1935 [hep-th]}}.

\bibitem{BoucherVeronneau:2011qv}
C.~Boucher-Veronneau and L.~J. Dixon, ``{N >- 4 Supergravity Amplitudes from
  Gauge Theory at Two Loops}'',
  \href{http://dx.doi.org/10.1007/JHEP12(2011)046}{{\em JHEP} {\bf 12} (2011)
  046},
\href{http://arxiv.org/abs/1110.1132}{{\tt arXiv:1110.1132 [hep-th]}}.

\bibitem{Bern:2012cd}
Z.~Bern, S.~Davies, T.~Dennen, and Y.-t. Huang, ``{Absence of Three-Loop
  Four-Point Divergences in N=4 Supergravity}'',
  \href{http://dx.doi.org/10.1103/PhysRevLett.108.201301}{{\em Phys. Rev.
  Lett.} {\bf 108} (2012)  201301},
\href{http://arxiv.org/abs/1202.3423}{{\tt arXiv:1202.3423 [hep-th]}}.

\bibitem{Bern:2012gh}
Z.~Bern, S.~Davies, T.~Dennen, and Y.-t. Huang, ``{Ultraviolet Cancellations in
  Half-Maximal Supergravity as a Consequence of the Double-Copy Structure}'',
  \href{http://dx.doi.org/10.1103/PhysRevD.86.105014}{{\em Phys. Rev.} {\bf
  D86} (2012)  105014},
\href{http://arxiv.org/abs/1209.2472}{{\tt arXiv:1209.2472 [hep-th]}}.

\bibitem{Bern:2012uf}
Z.~Bern, J.~J.~M. Carrasco, L.~J. Dixon, H.~Johansson, and R.~Roiban,
  ``{Simplifying Multiloop Integrands and Ultraviolet Divergences of Gauge
  Theory and Gravity Amplitudes}'',
  \href{http://dx.doi.org/10.1103/PhysRevD.85.105014}{{\em Phys. Rev.} {\bf
  D85} (2012)  105014},
\href{http://arxiv.org/abs/1201.5366}{{\tt arXiv:1201.5366 [hep-th]}}.

\bibitem{Du:2012mt}
Y.-J. Du and H.~Luo, ``{On General BCJ Relation at One-loop Level in Yang-Mills
  Theory}'', \href{http://dx.doi.org/10.1007/JHEP01(2013)129}{{\em JHEP} {\bf
  01} (2013)  129},
\href{http://arxiv.org/abs/1207.4549}{{\tt arXiv:1207.4549 [hep-th]}}.

\bibitem{Yuan:2012rg}
E.~Y. Yuan, ``{Virtual Color-Kinematics Duality: 6-pt 1-Loop MHV Amplitudes}'',
  \href{http://dx.doi.org/10.1007/JHEP05(2013)070}{{\em JHEP} {\bf 05} (2013)
  070},
\href{http://arxiv.org/abs/1210.1816}{{\tt arXiv:1210.1816 [hep-th]}}.

\bibitem{Bern:2013uka}
Z.~Bern, S.~Davies, T.~Dennen, A.~V. Smirnov, and V.~A. Smirnov, ``{Ultraviolet
  Properties of N=4 Supergravity at Four Loops}'',
  \href{http://dx.doi.org/10.1103/PhysRevLett.111.231302}{{\em Phys. Rev.
  Lett.} {\bf 111} (2013) no.~23, 231302},
\href{http://arxiv.org/abs/1309.2498}{{\tt arXiv:1309.2498 [hep-th]}}.

\bibitem{Boels:2013bi}
R.~H. Boels, R.~S. Isermann, R.~Monteiro, and D.~O'Connell,
  ``{Colour-Kinematics Duality for One-Loop Rational Amplitudes}'',
  \href{http://dx.doi.org/10.1007/JHEP04(2013)107}{{\em JHEP} {\bf 04} (2013)
  107},
\href{http://arxiv.org/abs/1301.4165}{{\tt arXiv:1301.4165 [hep-th]}}.

\bibitem{Bern:2013yya}
Z.~Bern, S.~Davies, T.~Dennen, Y.-t. Huang, and J.~Nohle, ``{Color-Kinematics
  Duality for Pure Yang-Mills and Gravity at One and Two Loops}'',
  \href{http://dx.doi.org/10.1103/PhysRevD.92.045041}{{\em Phys. Rev.} {\bf
  D92} (2015) no.~4, 045041},
\href{http://arxiv.org/abs/1303.6605}{{\tt arXiv:1303.6605 [hep-th]}}.

\bibitem{Bern:2013qca}
Z.~Bern, S.~Davies, and T.~Dennen, ``{The Ultraviolet Structure of Half-Maximal
  Supergravity with Matter Multiplets at Two and Three Loops}'',
  \href{http://dx.doi.org/10.1103/PhysRevD.88.065007}{{\em Phys. Rev.} {\bf
  D88} (2013)  065007},
\href{http://arxiv.org/abs/1305.4876}{{\tt arXiv:1305.4876 [hep-th]}}.

\bibitem{Bern:2014sna}
Z.~Bern, S.~Davies, and T.~Dennen, ``{Enhanced ultraviolet cancellations in
  $\mathcal N=5$ supergravity at four loops}'',
  \href{http://dx.doi.org/10.1103/PhysRevD.90.105011}{{\em Phys.Rev.} {\bf D90}
  (2014) no.~10, 105011},
\href{http://arxiv.org/abs/1409.3089}{{\tt arXiv:1409.3089 [hep-th]}}.

\bibitem{Mafra:2015mja}
C.~R. Mafra and O.~Schlotterer, ``{Two-loop five-point amplitudes of super
  Yang-Mills and supergravity in pure spinor superspace}'',
  \href{http://dx.doi.org/10.1007/JHEP10(2015)124}{{\em JHEP} {\bf 10} (2015)
  124},
\href{http://arxiv.org/abs/1505.02746}{{\tt arXiv:1505.02746 [hep-th]}}.

\bibitem{Johansson:2017bfl}
H.~Johansson, G.~K{\"a}lin, and G.~Mogull, ``{Two-loop supersymmetric QCD and
  half-maximal supergravity amplitudes}'',
  \href{http://dx.doi.org/10.1007/JHEP09(2017)019}{{\em JHEP} {\bf 09} (2017)
  019},
\href{http://arxiv.org/abs/1706.09381}{{\tt arXiv:1706.09381 [hep-th]}}.

\bibitem{Bern:2010yg}
Z.~Bern, T.~Dennen, Y.-t. Huang, and M.~Kiermaier, ``{Gravity as the Square of
  Gauge Theory}'', \href{http://dx.doi.org/10.1103/PhysRevD.82.065003}{{\em
  Phys.Rev.} {\bf D82} (2010)  065003},
\href{http://arxiv.org/abs/1004.0693}{{\tt arXiv:1004.0693 [hep-th]}}.

\bibitem{Cachazo:2013iea}
F.~Cachazo, S.~He, and E.~Y. Yuan, ``{Scattering of Massless Particles:
  Scalars, Gluons and Gravitons}'',
  \href{http://dx.doi.org/10.1007/JHEP07(2014)033}{{\em JHEP} {\bf 1407} (2014)
   033},
\href{http://arxiv.org/abs/1309.0885}{{\tt arXiv:1309.0885 [hep-th]}}.

\bibitem{Cachazo:2014xea}
F.~Cachazo, S.~He, and E.~Y. Yuan, ``{Scattering Equations and Matrices: From
  Einstein To Yang-Mills, DBI and NLSM}'',
  \href{http://dx.doi.org/10.1007/JHEP07(2015)149}{{\em JHEP} {\bf 07} (2015)
  149},
\href{http://arxiv.org/abs/1412.3479}{{\tt arXiv:1412.3479 [hep-th]}}.

\bibitem{Mason:2013sva}
L.~Mason and D.~Skinner, ``{Ambitwistor strings and the scattering
  equations}'', \href{http://dx.doi.org/10.1007/JHEP07(2014)048}{{\em JHEP}
  {\bf 1407} (2014)  048},
\href{http://arxiv.org/abs/1311.2564}{{\tt arXiv:1311.2564 [hep-th]}}.

\bibitem{Adamo:2013tsa}
T.~Adamo, E.~Casali, and D.~Skinner, ``{Ambitwistor strings and the scattering
  equations at one loop}'',
  \href{http://dx.doi.org/10.1007/JHEP04(2014)104}{{\em JHEP} {\bf 04} (2014)
  104},
\href{http://arxiv.org/abs/1312.3828}{{\tt arXiv:1312.3828 [hep-th]}}.

\bibitem{Monteiro:2014cda}
R.~Monteiro, D.~O'Connell, and C.~D. White, ``{Black holes and the double
  copy}'', \href{http://dx.doi.org/10.1007/JHEP12(2014)056}{{\em JHEP} {\bf
  1412} (2014)  056},
\href{http://arxiv.org/abs/1410.0239}{{\tt arXiv:1410.0239 [hep-th]}}.

\bibitem{Luna:2015paa}
A.~Luna, R.~Monteiro, D.~O'Connell, and C.~D. White, ``{The classical double
  copy for Taub--NUT spacetime}'',
  \href{http://dx.doi.org/10.1016/j.physletb.2015.09.021}{{\em Phys. Lett.}
  {\bf B750} (2015)  272--277},
\href{http://arxiv.org/abs/1507.01869}{{\tt arXiv:1507.01869 [hep-th]}}.

\bibitem{Luna:2016due}
A.~Luna, R.~Monteiro, I.~Nicholson, D.~O'Connell, and C.~D. White, ``{The
  double copy: Bremsstrahlung and accelerating black holes}'',
  \href{http://dx.doi.org/10.1007/JHEP06(2016)023}{{\em JHEP} {\bf 06} (2016)
  023},
\href{http://arxiv.org/abs/1603.05737}{{\tt arXiv:1603.05737 [hep-th]}}.

\bibitem{Alawadhi:2019urr}
R.~Alawadhi, D.~S. Berman, B.~Spence, and D.~Peinador~Veiga, ``{S-duality and
  the double copy}'', \href{http://dx.doi.org/10.1007/JHEP03(2020)059}{{\em
  JHEP} {\bf 03} (2020)  059}, \href{http://arxiv.org/abs/1911.06797}{{\tt
  arXiv:1911.06797 [hep-th]}}.

\bibitem{Banerjee:2019saj}
A.~Banerjee, E.~O. Colg\'ain, J.~A. Rosabal, and H.~Yavartanoo, ``{Ehlers as EM
  duality in the double copy}'',
  \href{http://dx.doi.org/10.1103/PhysRevD.102.126017}{{\em Phys. Rev. D} {\bf
  102} (2020)  126017}, \href{http://arxiv.org/abs/1912.02597}{{\tt
  arXiv:1912.02597 [hep-th]}}.

\bibitem{Luna:2018dpt}
A.~Luna, R.~Monteiro, I.~Nicholson, and D.~O'Connell, ``{Type D Spacetimes and
  the Weyl Double Copy}'',
  \href{http://dx.doi.org/10.1088/1361-6382/ab03e6}{{\em Class. Quant. Grav.}
  {\bf 36} (2019)  065003}, \href{http://arxiv.org/abs/1810.08183}{{\tt
  arXiv:1810.08183 [hep-th]}}.

\bibitem{White:2020sfn}
C.~D. White, ``{A Twistorial Foundation for the Classical Double Copy}'',
  \href{http://arxiv.org/abs/2012.02479}{{\tt arXiv:2012.02479 [hep-th]}}.

\bibitem{Monteiro:2020plf}
R.~Monteiro, D.~O'Connell, D.~P. Veiga, and M.~Sergola, ``{Classical Solutions
  and their Double Copy in Split Signature}'',
  \href{http://arxiv.org/abs/2012.11190}{{\tt arXiv:2012.11190 [hep-th]}}.

\bibitem{Chacon:2021wbr}
E.~Chac\'on, S.~Nagy, and C.~D. White, ``{The Weyl double copy from twistor
  space}'', \href{http://arxiv.org/abs/2103.16441}{{\tt arXiv:2103.16441
  [hep-th]}}.

\bibitem{Bern:2019crd}
Z.~Bern, C.~Cheung, R.~Roiban, C.-H. Shen, M.~P. Solon, and M.~Zeng, ``{Black
  Hole Binary Dynamics from the Double Copy and Effective Theory}'',
\href{http://arxiv.org/abs/1908.01493}{{\tt arXiv:1908.01493 [hep-th]}}.

\bibitem{Bern:2019nnu}
Z.~Bern, C.~Cheung, R.~Roiban, C.-H. Shen, M.~P. Solon, and M.~Zeng,
  ``{Scattering Amplitudes and the Conservative Hamiltonian for Binary Systems
  at Third Post-Minkowskian Order}'',
  \href{http://dx.doi.org/10.1103/PhysRevLett.122.201603}{{\em Phys. Rev.
  Lett.} {\bf 122} (2019) no.~20, 201603},
\href{http://arxiv.org/abs/1901.04424}{{\tt arXiv:1901.04424 [hep-th]}}.

\bibitem{Bern:2020buy}
Z.~Bern, A.~Luna, R.~Roiban, C.-H. Shen, and M.~Zeng, ``{Spinning Black Hole
  Binary Dynamics, Scattering Amplitudes and Effective Field Theory}'',
  \href{http://arxiv.org/abs/2005.03071}{{\tt arXiv:2005.03071 [hep-th]}}.

\bibitem{Bern:2020gjj}
Z.~Bern, H.~Ita, J.~Parra-Martinez, and M.~S. Ruf, ``{Universality in the
  classical limit of massless gravitational scattering}'',
  \href{http://arxiv.org/abs/2002.02459}{{\tt arXiv:2002.02459 [hep-th]}}.

\bibitem{Bern:2020uwk}
Z.~Bern, J.~Parra-Martinez, R.~Roiban, E.~Sawyer, and C.-H. Shen, ``{Leading
  Nonlinear Tidal Effects and Scattering Amplitudes}'',
  \href{http://arxiv.org/abs/2010.08559}{{\tt arXiv:2010.08559 [hep-th]}}.

\bibitem{Bern:2021dqo}
Z.~Bern, J.~Parra-Martinez, R.~Roiban, M.~S. Ruf, C.-H. Shen, M.~P. Solon, and
  M.~Zeng, ``{Scattering Amplitudes and Conservative Binary Dynamics at ${\cal
  O}(G^4)$}'', \href{http://arxiv.org/abs/2101.07254}{{\tt arXiv:2101.07254
  [hep-th]}}.

\bibitem{Cheung:2020djz}
C.~Cheung and J.~Mangan, ``{Scattering Amplitudes and the Navier-Stokes
  Equation}'', \href{http://arxiv.org/abs/2010.15970}{{\tt arXiv:2010.15970
  [hep-th]}}.

\bibitem{Campiglia:2021srh}
M.~Campiglia and S.~Nagy, ``{A double copy for asymptotic symmetries in the
  self-dual sector}'', \href{http://arxiv.org/abs/2102.01680}{{\tt
  arXiv:2102.01680 [hep-th]}}.

\bibitem{Tolotti:2013caa}
M.~Tolotti and S.~Weinzierl, ``{Construction of an effective Yang-Mills
  Lagrangian with manifest BCJ duality}'',
  \href{http://dx.doi.org/10.1007/JHEP07(2013)111}{{\em JHEP} {\bf 07} (2013)
  111},
\href{http://arxiv.org/abs/1306.2975}{{\tt arXiv:1306.2975 [hep-th]}}.

\bibitem{Ferrero:2020vww}
P.~Ferrero and D.~Francia, ``{On the Lagrangian formulation of the double copy
  to cubic order}'', \href{http://dx.doi.org/10.1007/JHEP02(2021)213}{{\em
  JHEP} {\bf 02} (2021)  213}, \href{http://arxiv.org/abs/2012.00713}{{\tt
  arXiv:2012.00713 [hep-th]}}.

\bibitem{LopesCardoso:2018xes}
G.~Lopes~Cardoso, G.~Inverso, S.~Nagy, and S.~Nampuri, ``{Comments on the
  double copy construction for gravitational theories}'', in {\em {17th
  Hellenic School and Workshops on Elementary Particle Physics and Gravity
  (CORFU2017) Corfu, Greece, September 2-28, 2017}}.
\newblock 2018.
\newblock \href{http://arxiv.org/abs/1803.07670}{{\tt arXiv:1803.07670
  [hep-th]}}.
\newblock
\url{http://inspirehep.net/record/1663475/files/1803.07670.pdf}.
\newblock

\bibitem{Borsten:2020zgj}
L.~Borsten, B.~Jurco, H.~Kim, T.~Macrelli, C.~Saemann, and M.~Wolf,
  ``{BRST-Lagrangian Double Copy of Yang-Mills Theory}'',
  \href{http://arxiv.org/abs/2007.13803}{{\tt arXiv:2007.13803 [hep-th]}}.

\bibitem{Borsten:2021hua}
L.~Borsten, B.~Jurco, H.~Kim, T.~Macrelli, C.~Saemann, and M.~Wolf, ``{Double
  Copy from Homotopy Algebras}'', \href{http://arxiv.org/abs/2102.11390}{{\tt
  arXiv:2102.11390 [hep-th]}}.

\bibitem{Adamo:2017nia}
T.~Adamo, E.~Casali, L.~Mason, and S.~Nekovar, ``{Scattering on plane waves and
  the double copy}'', \href{http://dx.doi.org/10.1088/1361-6382/aa9961}{{\em
  Class. Quant. Grav.} {\bf 35} (2018) no.~1, 015004},
\href{http://arxiv.org/abs/1706.08925}{{\tt arXiv:1706.08925 [hep-th]}}.

\bibitem{Adamo:2020qru}
T.~Adamo and A.~Ilderton, ``{Classical and quantum double copy of
  back-reaction}'', \href{http://dx.doi.org/10.1007/JHEP09(2020)200}{{\em JHEP}
  {\bf 09} (2020)  200}, \href{http://arxiv.org/abs/2005.05807}{{\tt
  arXiv:2005.05807 [hep-th]}}.

\bibitem{Bahjat-Abbas:2017htu}
N.~Bahjat-Abbas, A.~Luna, and C.~D. White, ``{The Kerr-Schild double copy in
  curved spacetime}'', \href{http://dx.doi.org/10.1007/JHEP12(2017)004}{{\em
  JHEP} {\bf 12} (2017)  004},
\href{http://arxiv.org/abs/1710.01953}{{\tt arXiv:1710.01953 [hep-th]}}.

\bibitem{Alkac:2021bav}
G.~Alkac, M.~K. Gumus, and M.~Tek, ``{The Classical Double Copy in Curved
  Spacetime}'', \href{http://arxiv.org/abs/2103.06986}{{\tt arXiv:2103.06986
  [hep-th]}}.

\bibitem{Farrow:2018yni}
J.~A. Farrow, A.~E. Lipstein, and P.~McFadden, ``{Double copy structure of CFT
  correlators}'', \href{http://dx.doi.org/10.1007/JHEP02(2019)130}{{\em JHEP}
  {\bf 02} (2019)  130},
\href{http://arxiv.org/abs/1812.11129}{{\tt arXiv:1812.11129 [hep-th]}}.

\bibitem{Lipstein:2019mpu}
A.~Lipstein and P.~McFadden, ``{Double copy structure and the flat space limit
  of conformal correlators in even dimensions}'',
\href{http://arxiv.org/abs/1912.10046}{{\tt arXiv:1912.10046 [hep-th]}}.

\bibitem{Armstrong:2020woi}
C.~Armstrong, A.~E. Lipstein, and J.~Mei, ``{Color/kinematics duality in
  AdS$_{4}$}'', \href{http://dx.doi.org/10.1007/JHEP02(2021)194}{{\em JHEP}
  {\bf 02} (2021)  194}, \href{http://arxiv.org/abs/2012.02059}{{\tt
  arXiv:2012.02059 [hep-th]}}.

\bibitem{Albayrak:2020fyp}
S.~Albayrak, S.~Kharel, and D.~Meltzer, ``{On duality of color and kinematics
  in (A)dS momentum space}'',
  \href{http://dx.doi.org/10.1007/JHEP03(2021)249}{{\em JHEP} {\bf 03} (2021)
  249}, \href{http://arxiv.org/abs/2012.10460}{{\tt arXiv:2012.10460
  [hep-th]}}.
  
\bibitem{alday2021gluonp}
L. F. Alday, C. Behan, P. Ferrero, and X. Zhou, ``{Gluon Scattering in AdS from CFT}'', \href{http://arxiv.org/abs/2103.15830}{{\tt arXiv:2103.15830
  [hep-th]}}.

\bibitem{Helgason:2001dif}
S.~Helgason, {\em Differential geometry and symmetric spaces}, vol.~341.
\newblock American Mathematical Soc., 2001.

\bibitem{Farashahi:2013con}
A.~G. Farashahi, ``Convolution and involution on function spaces of homogeneous
  spaces'', {\em Bulletin of the Malaysian Mathematical Sciences Society} {\bf
  36} (2013) no.~4, , \href{http://arxiv.org/abs/1201.0297}{{\tt 1201.0297}}.

\bibitem{Farashahi:2015}
A.~G. Farashahi, ``{Abstract convolution function algebras over homogeneous
  spaces of compact groups}'',
  \href{http://dx.doi.org/10.1215/ijm/1488186019}{{\em Illinois Journal of
  Mathematics} {\bf 59} (2015) no.~4, 1025 -- 1042}.
  \url{https://doi.org/10.1215/ijm/1488186019}.

\bibitem{Luna:2016hge}
A.~Luna, R.~Monteiro, I.~Nicholson, A.~Ochirov, D.~O'Connell, N.~Westerberg,
  and C.~D. White, ``{Perturbative spacetimes from Yang-Mills theory}'',
  \href{http://dx.doi.org/10.1007/JHEP04(2017)069}{{\em JHEP} {\bf 04} (2017)
  069},
\href{http://arxiv.org/abs/1611.07508}{{\tt arXiv:1611.07508 [hep-th]}}.

\bibitem{peterweyl}
F. Peter, and H. Weyl, ``{Die Vollst\"{a}ndigkeit der primitiven Darstellungen einer geschlossenen kontinuierlichen Gruppe}'',
  \href{https://doi.org/10.1007/BF01447892}{{\em Math. Ann. } {\bf 97} (1927), 737 -- 755}.
  \url{https://doi.org/10.1007/BF01447892}.
  
\bibitem{gallot2004riemannian}
S. Gallot, D. Hulin, and J. Lafontaine, ``{Riemannian Geometry}'',
  \href{https://doi.org/10.1007/BF01447892}{{\em Springer-Verlag Berlin Heidelberg} (2004)}.
  \url{https://doi.org/10.1007/978-3-642-18855-8}.
  
 

\end{thebibliography}
\providecommand{\href}[2]{#2}\begingroup\raggedright\endgroup

\end{document}